\documentclass[twocolumn]{article}
\usepackage{amsthm,amssymb,amsmath,graphicx}
\usepackage{color}
\usepackage[top=2cm, bottom=2cm, left=2.5cm, right=3cm]{geometry}
\usepackage[pagebackref=false,colorlinks,linkcolor=blue,citecolor=magenta]{hyperref}
\usepackage{multicol}
\usepackage{tikz}
\usepackage{algcompatible}
\usepackage{subcaption}
\usepackage{algorithm}
\usepackage{algpseudocode}
\usepackage{tcolorbox}
\usetikzlibrary{arrows.meta, positioning,calc}
\usepackage{lipsum}   

\newcommand{\ket}[1]{|#1\rangle}
\newcommand{\bra}[1]{\langle #1|}
\newcommand{\proj}[1]{\ket{#1}\!\bra{#1}}
\newcommand{\ketbra}[2]{\ket{#1}\!\bra{#2}}

\usepackage[english]{babel}

\theoremstyle{plain}
\newtheorem{theorem}{Theorem}[section]

\newtheorem{lemma}[theorem]{Lemma}

\newtheorem{proposition}[theorem]{Proposition}

\theoremstyle{definition}
\newtheorem{example}[theorem]{Example}
\newtheorem{definition}[theorem]{Definition}

\theoremstyle{remark}
\newtheorem{remark}[theorem]{Remark}

\usepackage{amsmath}
\DeclareMathOperator*{\argmax}{arg\,max}
\DeclareMathOperator*{\argmin}{arg\,min}

\newcommand{\la}{\langle}
\newcommand{\ra}{\rangle}
\newcommand{\mb}[1]{\mathbb{#1}}
\newcommand{\mc}[1]{\mathcal{#1}}
\newcommand{\tr}{\operatorname{tr}}
\newcommand{\Tr}{\operatorname{Tr}}

\newcommand{\supp}{\text{supp}}

\newcommand{\rank}{\text{rank}}

\newcommand{\cA}{\mc A}
\newcommand{\cB}{\mc B}
\newcommand{\cT}{\mc T}
\newcommand{\cE}{\mc E}
\newcommand{\cR}{\mc R}
\newcommand{\cM}{\mc M}
\newcommand{\cD}{\mc{D}}
\newcommand{\cS}{\mc{S}}

\newcommand{\floor}[1]{\left\lfloor #1 \right\rfloor}
\usepackage[affil-it]{authblk}

\title{Capacities of quantum Markovian noise for large times}

\author{Omar Fawzi}
\author{Mizanur Rahaman}
\author{Mostafa Taheri}

\affil[1]{Univ Lyon, Inria, ENS Lyon, UCBL, LIP, Lyon, France}

\begin{document}

\maketitle
\begin{abstract}
Given a quantum Markovian noise model, we study the maximum dimension of a classical or quantum system that can be stored for arbitrarily large time. We show that, unlike the fixed time setting, in the limit of infinite time, the classical and quantum capacities are characterized by efficiently computable properties of the peripheral spectrum of the quantum channel. In addition, the capacities are additive under tensor product, which implies in the language of Shannon theory that the one-shot and the asymptotic i.i.d. capacities are the same. We also provide an improved algorithm for computing the structure of the peripheral subspace of a quantum channel, which might be of independent interest.
\end{abstract}

\section{Introduction }

Consider a quantum system that we would like to use for storing quantum or classical information. This system is affected by noise that we assume is Markovian. It is natural to ask what is the minimum error that can be achieved for storing $\log D$ (qu)bits of information for some fixed time $t$. This is a typical question studied in Shannon theory, but here we focus on the limit $t \to \infty$, i.e., the information should remain in the system for arbitrarily long times.

 Building a quantum system able to store quantum information for large times is one of the important goals of quantum information theory and it has been studied from different angles. Quantum error correction gives a mechanism to actively preserve the quantum information in a system undergoing local noise~\cite{terhal2015quantum}. Such methods can achieve much more than a memory and can be used for fault-tolerant quantum computation~\cite{gottesman2016surviving}. A related important area of research is the study of self-correcting (or passive) quantum memories~\cite{brown2016quantum}, i.e., physical systems that are robust to different forms of imperfections including thermal noise.
 
 In this paper, we study the question of quantum memory from an abstract perspective where given a dynamical quantum system, the objective is to characterize the maximum amount of information that can be stored in this system without placing restrictions on the encoding/decoding operations.  
  
 As an example, in the passive model, the noise $\cT$ is applied at each time step and we would like to characterize how many qubits can be stored reliably for time $t$ as a function of $\cT$ and $t$? Our framework can also model the active setting where a fixed recovery operation $\cR$ is applied at each time step; see Fig~\ref{fig:two_column_diagram} for an illustration. 
 Note that in our modeling, $\cT$ is not limited to representing undesirable noise; it can also be a model for  an engineered system such as cat qubits (see Examples~\ref{ex: cat code} and~\ref{ex: noise+ recovery}  for a discussion).
 Such questions were studied in the recent work~\cite{singh2024-v1} and they obtained, among other results about scrambling, conditions for the classical capacity to be zero as well as lower bounds on the classical capacity of ergodic channels.
 \begin{figure}[t]
\centering
\resizebox{8cm}{!}{

\begin{tikzpicture}[>=Stealth, every node/.style={align=center}]
    \node[rectangle, draw, fill=purple!20, minimum width=0.8cm, minimum height=0.9cm] (E) at (0, 0) {\(\mathcal{E}\)};

    \node[rectangle, draw, minimum width=2cm, minimum height=1.2cm, right=0.8cm of E] (T1) {};
    \node[rectangle, draw, fill=green!20, minimum width=0.7cm, minimum height=0.9cm, anchor=west] (N1) at ($(T1.west)+(0.15,0)$) {\(\mathcal{N}\)};
    \node[rectangle, draw, fill=blue!20, minimum width=0.7cm, minimum height=0.9cm, anchor=east] (R1) at ($(T1.east)-(0.15,0)$) {\(\mathcal{R}\)};
    \node at (T1.north) [above=0.2cm] {\(\mathcal{T}\)};
    \draw[->] (N1.east) -- (R1.west);

    \node[rectangle, draw, minimum width=2cm, minimum height=1.2cm, right=0.8cm of T1] (T2) {};
    \node[rectangle, draw, fill=green!20, minimum width=0.7cm, minimum height=0.9cm, anchor=west] (N2) at ($(T2.west)+(0.15,0)$) {\(\mathcal{N}\)};
    \node[rectangle, draw, fill=blue!20, minimum width=0.7cm, minimum height=0.9cm, anchor=east] (R2) at ($(T2.east)-(0.15,0)$) {\(\mathcal{R}\)};
    \node at (T2.north) [above=0.2cm] {\(\mathcal{T}\)};
    \draw[->] (N2.east) -- (R2.west);

    \node[circle, draw, fill=black, inner sep=0pt, minimum size=1.5pt, right=0.8cm of T2] (dot1) {};
    \node[circle, draw, fill=black, inner sep=0pt, minimum size=1.5pt, right=0.3cm of dot1] (dot2) {};
    \node[circle, draw, fill=black, inner sep=0pt, minimum size=1.5pt, right=0.3cm of dot2] (dot3) {};

    \node[rectangle, draw, minimum width=2cm, minimum height=1.2cm, right=0.8cm of dot3] (T3) {};
    \node[rectangle, draw, fill=green!20, minimum width=0.7cm, minimum height=0.9cm, anchor=west] (N3) at ($(T3.west)+(0.15,0)$) {\(\mathcal{N}\)};
    \node[rectangle, draw, fill=blue!20, minimum width=0.7cm, minimum height=0.9cm, anchor=east] (R3) at ($(T3.east)-(0.15,0)$) {\(\mathcal{R}\)};
    \node at (T3.north) [above=0.2cm] {\(\mathcal{T}\)};
    \draw[->] (N3.east) -- (R3.west);

    \node[rectangle, draw, fill=orange!20, minimum width=0.8cm, minimum height=0.9cm, right=0.8cm of T3] (D) {\(\mathcal{D}\)};

    \draw[->] ($(E.east) + (0,0)$) -- ($(T1.west) - (0.2,0)$);
    \draw[->] ($(T1.east) + (0.2,0)$) -- ($(T2.west) - (0.2,0)$);
    \draw[->] ($(T2.east) + (0.2,0)$) -- ($(dot1.west) - (0.2,0)$);
    \draw[->] ($(dot3.east) + (0.2,0)$) -- ($(T3.west) - (0.2,0)$);
    \draw[->] ($(T3.east) + (0.2,0)$) -- ($(D.west) - (0.2,0)$);

    \node[rectangle, draw, fill=purple!20, minimum width=0.8cm, minimum height=0.9cm, above=1.5cm of E] (UE) {\(\mathcal{E}\)};

    \node[rectangle, draw, fill=cyan!20, minimum width=2cm, minimum height=1.2cm, right=0.8cm of UE] (UT1) {\(\mathcal{T}\)};
    \node[rectangle, draw, fill=cyan!20, minimum width=2cm, minimum height=1.2cm, right=0.8cm of UT1] (UT2) {\(\mathcal{T}\)};
    \node[circle, draw, fill=black, inner sep=0pt, minimum size=1.5pt, right=0.8cm of UT2] (UDot1) {};
    \node[circle, draw, fill=black, inner sep=0pt, minimum size=1.5pt, right=0.3cm of UDot1] (UDot2) {};
    \node[circle, draw, fill=black, inner sep=0pt, minimum size=1.5pt, right=0.3cm of UDot2] (UDot3) {};
    \node[rectangle, draw, fill=cyan!20, minimum width=2cm, minimum height=1.2cm, right=0.8cm of UDot3] (UT3) {\(\mathcal{T}\)};
    \node[rectangle, draw, fill=orange!20, minimum width=0.8cm, minimum height=0.9cm, right=0.8cm of UT3] (UD) {\(\mathcal{D}\)};

    \draw[->] ($(UE.east) + (0,0)$) -- ($(UT1.west) - (0.2,0)$);
    \draw[->] ($(UT1.east) + (0.2,0)$) -- ($(UT2.west) - (0.2,0)$);
    \draw[->] ($(UT2.east) + (0.2,0)$) -- ($(UDot1.west) - (0.2,0)$);
    \draw[->] ($(UDot3.east) + (0.2,0)$) -- ($(UT3.west) - (0.2,0)$);
    \draw[->] ($(UT3.east) + (0.2,0)$) -- ($(UD.west) - (0.2,0)$);

    \node[above left=0.3cm of UE] {(a)};
    \node[above left=0.3cm of E] {(b)};

    \draw[thick, ->] ($(E.west)-(.3,1)$) -- ($(D.east)+(0,-1)$) node[midway, below] {\textit{time}};
\end{tikzpicture}
}
\caption{
(a) \textbf{Passive error correction}: The system evolves under noise (\(\mathcal{T}\)) over time, starting with encoding (\(\mathcal{E}\)) and ending with decoding (\(\mathcal{D}\)).  
(b) \textbf{Active error correction}: Noise (\(\mathcal{N}\)) and recovery (\(\mathcal{R}\)) maps are applied periodically to maintain system integrity, from encoding (\(\mathcal{E}\)) to decoding (\(\mathcal{D}\)).
}\label{fig:two_column_diagram}
\end{figure}

\paragraph{Our results} In this work, we focus on the setting of arbitrarily large time, i.e., $t \to \infty$ and we characterize both the classical and quantum capacities for a fixed error $\delta$ in terms of the peripheral spectrum of the noise model (Theorem~\ref{thm:classical and quantum capacity of repeated channel}). In addition, we show that such capacities are additive for the tensor product of channels (Theorem~\ref{thm:additivity}) and can be computed in polynomial-time in the dimension (Theorem~\ref{thm:algorithm-structure}). This algorithm for computing the structure of the peripheral subspace (or the fixed point subspace) of a quantum channel improves on previous works and might of independent interest. We note that the fact that we can efficiently characterize the infinite-time capacities of noisy channels is in contrast with other settings for which capacities or optimal success probabilities correspond to hard problems, such as maximum independent set~\cite{shannon1956zero} or maximum coverage problems~\cite{barman2017algorithmic}.

Our analysis extends to infinite-dimensional Hilbert spaces, where we identify specific conditions under which our findings remain valid (see Proposition~\ref{prop: infinte-hilbert-capacity}). This generalization broadens the applicability of our results
 
to Markovian noise on continuous variable systems. We then illustrate our framework through examples in Section~\ref{sec: examples}.

\section{Results}
For a Hilbert space $ H$, we use the notation $\Tr( H)$ for the trace class operators on $H$. When we do not need to make the Hilbert space explicit and when it has dimension $d$, we denote the algebra of linear operators (or $d \times d$ complex matrices) by $\mc M_d(\mb{C})$ or simply $\mc M_d$ for short. We recall that a quantum channel $\mc E : \mc M_D \to \mc M_d$ is a completely positive and trace preserving linear map~\cite{Wolfe2012}.

We start with the standard definition of an error correcting code for a noisy quantum channel $\mc T$. Note that throughout the paper $d \geq 1$ is an integer and $\mc T$ is a quantum channel from $\cM_d$ to $\cM_d$.

A quantum channel $\mc E: \mc M_D(\mb C)\to \mc M_d(\mb C) $ is a $(D,\delta)$ classical code for a channel $\mc T: \mc M_d(\mb C)\to \mc M_d(\mb C)$ if there exists a recovery channel $\mc R: \mc M_d(\mb C)\to \mc M_D(\mb C)$ such that the average fidelity of the channel $\mc R\circ \mc T\circ \mc E$ over all diagonal density matrices in $\mc M_D$ is at least $1-\delta$, i.e.,
        \begin{equation}
            \frac 1D \sum_{i=1}^D \la i| \mc R\circ \mc T\circ \mc E(|i\ra\la i|) |i\ra \geq 1-\delta,
        \end{equation}
        where $\{\ket{i}\}_{i=1}^D$ is a fixed orthonormal basis of $\mb C^D$. We say that $\mc E$ is a $(D,\delta)$ quantum code if there exist channel $\mc R$ such that the entanglement fidelity of $\mc R\circ \mc T\circ \mc E$ is at least $1-\delta$, i.e.,
        \begin{equation}
             \bra{\Phi^+} \left(\mc{I} \otimes (\mc R\circ \mc T\circ \mc E) \right)\left(|\Phi^+\ra\la \Phi^+|\right) \ket{\Phi^+} \geq 1-\delta,
        \end{equation}
        where $|\Phi^+\ra = \frac{1}{\sqrt{D}} \sum_{i=1}^D \ket{i} \otimes \ket{i} \in (\mb C^D)^{\otimes 2}$ and $\mc{I}$ is the identity quantum channel on the reference system of dimension $D$.
 \begin{definition}[Capacity]
        The (one-shot) classical capacity for a channel $\mc{T}$ is defined as 
                $$C_\delta(\mc{T}) = \sup_{\mc{E}} \log D,$$
        where the supremum is over all possible $(D,\delta)$ classical codes $\mc{E}$ for $\mc{T}$. Similarly, the (one-shot) quantum capacity for the channel $\mc{T}$ is defined as 
        $$Q_\delta(\mc{T}) = \sup_{\mc{E}} \log D,$$ 
        where the supremum is taken over all possible $(D,\delta)$ quantum codes.
    \end{definition}

 \begin{remark}
    This definition is concerned about passive error correction where no recovery is allowed at regular intervals. One could define an active capacity as follows:
    \begin{align*}
         Q_\delta^{\mathrm{active}, t}(\mc T) &= \log \{\max D : \exists \cE, \cR_1, \dots, \cR_{t} 
         \text{ s.t. }\\ & F_E(\cR_t \circ \cT \circ \cdots \circ \cT \circ \cR_1 \circ \cT \circ \cE) \geq 1-\delta \},
    \end{align*}
   
    where $\cE : \cM_D \to \cM_d$ and $\cR_i : \cM_d \to \cM_d$ for $1 \leq i \leq t-1$ and $\cR_t : \cM_d \to \cM_D$ are quantum channels. Note that the active capacity can be larger than the passive one; in particular, it is simple to see that for any $t \geq 1$, $Q_0^{\mathrm{active}, t}(\mc T) = Q_0(\mc T)$ by choosing the recovery maps $\cR_i$ to decode and re-encode information. The works~\cite{muller2014quantum,fawzi2022lower,pirandola2019end} have studied some Shannon-theoretic aspects of active error correction.
    \end{remark}

        The \textbf{infinite-time classical or quantum capacity} of a quantum channel $\mc T$ is defined as
        \( C_\delta^\infty (\mc T)= \lim_{t\to \infty} C_\delta(\mc T^t)\), and $  Q_\delta^\infty(\mc T) = \lim_{t\to \infty} Q_\delta(\mc T^t). $

As we will demonstrate in the following, the peripheral subspace of a channel is directly related to its infinite-time capacity. The peripheral subspace of a quantum channel $\mathcal{T}$ is defined as follows:
\[
\chi_{\mathcal{T}} =\text{span} \{ X \in \mathcal{M}_d \mid \mathcal{T}(X) = \lambda X, \, |\lambda| = 1, \, \lambda \in \mathbb{C} \}.
\]

 For peripheral subspace $\chi_\mc T$ there exists  Hilbert space decomposition $\mb{C}^d = H_0 \oplus \bigoplus_{k=1}^K (H_{k,1} \otimes H_{k,2})$  such that

    \begin{align}
    \label{eq:peripheral-decomposition-1}
        \chi_{\mc{T}} = 0 \oplus \bigoplus_{k=1}^K \mc{M}_{d_k} \otimes \omega_k,
    \end{align}
    where $\mc{M}_{d_k}$ is the full matrix algebra on $H_{k,1}$, with $d_k = \dim H_{k,1}$, and $\omega_k$ is a density operator on $H_{k,2}$~\cite{Wolfe2012}.

\begin{theorem}\label{thm:classical and quantum capacity of repeated channel}
        Let $\mc{T}: \mc M_d \rightarrow \mc M_ d$ be a quantum channel and let $K$ and $\{d_k\}_{k=1}^K$ be the integers from the peripheral subspace decomposition in~\eqref{eq:peripheral-decomposition}.
        
        For any $\delta \in [0,1)$, the infinite time classical capacity of $\mc T$ is given by
        \begin{equation}    
        \label{eq:classical-cap-formula}
         C_{\delta}^\infty(\mc{T}) = \log
        \left(\left\lfloor \frac{\sum_k d_k}{1-\delta} \right\rfloor \right),
        \end{equation}
        and the quantum capacity can be bounded as
        \begin{equation}
        \label{eq:quantum-cap-bound}
        \log\left(\floor{\frac{\max_k d_k}{\sqrt{1-\delta}}}\right) \leq  Q_{\delta}^\infty (\mc T) \leq \log\left(\frac {\max_k d_k}{1-\delta}\right).
        \end{equation}

    \end{theorem}
    We note that, in independent work, the special case $\delta = 0$ of \eqref{eq:classical-cap-formula} and \eqref{eq:quantum-cap-bound} were proved in~\cite{singh2024-v2}.
    
    Continuous-time Markovian noise $\mc T_t$ are known as quantum Markov semigroups (QMS) generated by a Lindblad operator $\mathcal{L}$, i.e., $\mc T_t = e^{\mc L t}$~\cite{rivas2012open,GKS,lindblad1976generators}.
    Similarly, the results presented in Theorem~\ref{thm:classical and quantum capacity of repeated channel} hold for a QMS. Furthermore, the peripheral subspace of a QMS for any time can be expressed in terms of the spectrum of the generator $\mathcal{L}$ as follows:
\[
\chi_{e^{\mathcal{L} t}} = \text{span}\{ X \in \mathcal{M}_d \mid \exists \theta \in \mathbb{R}, \, \mathcal{L}(X) = i\theta X \}.
\]
For more details check Proposition~\ref{prop: Markov semigroup capacities}.

One of the important questions regarding any capacity is additivity under tensor product. Since the infinite-time capacity is determined by the peripheral subspace, it is essentially additive under the tensor product.
    
 \begin{theorem}
    \label{thm:additivity}
        For any two quantum channels $\mc T$ and $\mc S$, the infinite-time zero-error classical and quantum capacities are additive under the tensor product. Specifically, we have:
        \begin{align*}
             C_{0}^\infty\left(\mc T \otimes \mc S\right) &=  C_{0}^\infty\left(\mc T\right) +  C_{0}^\infty\left(\mc S\right), \\
              Q_{0}^\infty\left(\mc T \otimes \mc S\right) &=  Q_{0}^\infty\left(\mc T\right) + Q_{0}^\infty\left(\mc S\right).
        \end{align*}
       
    \end{theorem}

   For the asymptotic scenario, considering the limit of the infinite tensor product of a channel, the following results hold:
\begin{proposition}\label{prop: Shannon theory capacity}
        Let $\delta \in [0,1)$ and $\mc T: \mc M_d \to \mc M_d$ be a quantum channel. Then 
        \[ \lim_{m\to \infty}  \frac 1m C_{\delta}^\infty \left(\mc T^{\otimes m}\right)= C_0^\infty(\mc T),\]
        and
        \[ \lim_{m\to \infty}  \frac 1m Q_{\delta}^\infty\left(\mc T^{\otimes m}\right)= Q_0^\infty(\mc T).\]
    \end{proposition}

To compute the infinite-time capacities of a given channel $\mathcal{T}$, it is crucial to determine the structure of its peripheral subspace. An algorithm is proposed in~\cite{blume2010information} that, in polynomial time, maps the structure of the peripheral subspace to the structure of a finite-dimensional von Neumann algebra. Consequently, our algorithm can also be utilized for determining the structure of finite-dimensional von Neumann algebras.
 
\begin{theorem}
    \label{thm:algorithm-structure}
        Let \(\mathcal{T}: \mathcal{M}_d \to \mathcal{M}_d\) be a quantum channel, and let \(\chi_{\mathcal{T}} = 0 \oplus \bigoplus_{k=1}^K \mathcal{M}_{d_k}\otimes \omega_k\) be the decomposition of its peripheral subspace as described in  \eqref{eq:peripheral-decomposition-1}.
        Algorithm~\ref{alg: Find the structure of chi}, takes the super-operator form of \(\mathcal{T}\) as input and returns a representation of the Hilbert space decomposition and the fixed density matrices \(\omega_k\) in time \(O(d^6 \log d)\).

    \end{theorem}
    \begin{remark}
        An implementation of this algorithm with examples can be found in~\cite{mostafa_taheri_2024_12821315}.
    \end{remark}

\subsection{Infinite-dimensional Hilbert spaces}
So far, we have focused on finite-dimensional Hilbert spaces. Now, let us shift our attention to the case of infinite-dimensional separable Hilbert spaces. Some concepts from finite-dimensional spaces do not extend directly to infinite dimensions, introducing new challenges and subtleties. To illustrate the differences in infinite-dimensional spaces, let us consider an example:

\begin{example}
    
    Let \(\{ |i\rangle \}_{i \in \mathbb{Z}}\) be an orthonormal basis for the Hilbert space \(H\), and let \(\Tr(H)\) denote the trace class space, i.e., the space of compact operators on \(H\) with finite trace norm. Now, consider the quantum channel \(\mathcal{T}: \Tr(H) \to \Tr(H)\) defined by 
    \[
    \mathcal{T}(X) = UXU^\dagger,
    \]
    where \(U = \sum_{i\in \mb Z} |i+1\rangle \langle i|\) is known as the bilateral shift operator. This operator shifts each basis state \(|i\rangle\) to \(|i+1\rangle\), effectively ``shifting the indices up by one."

    It is easy to see that the channel has an empty peripheral subspace~\cite{olkiewicz_1999}.  Despite the absence of a peripheral subspace, this channel preserves the distinguishability of input states, meaning that \(\tr(\rho \sigma) = \tr(\mathcal{T}(\rho) \mathcal{T}(\sigma))\), and thus its capacity is infinite despite the lack of a peripheral subspace.
\end{example}

    The example above makes it clear that the peripheral subspace is not, in general, the relevant subspace for determining long-time capacity in the infinite-dimensional case. 
   
    Instead, for an infinite-dimensional system, the isometric subspace of the channel \(\mathcal{T}\), denoted \(\Lambda_{\mathcal{T}}\), provides suitable properties~\cite{olkiewicz_1999}.
     This isometric subspace is defined as
    \begin{align*}
        \Lambda_{\mathcal{T}} := \{& x \in \operatorname{Tr}(H) :  \forall t \in \mathbb{N},\\
        &\| \mathcal{T}^t(x) \|_2 = \| \mathcal{T}^{*t}(x) \|_2 = \| x \|_2  \}.
    \end{align*}
   
\begin{remark}
    If \(\mc{T}\) is a quantum channel on a finite-dimensional Hilbert space, then its isometric subspace coincides with its peripheral subspace~\cite{olkiewicz_1999}.
\end{remark}
\begin{proposition}
    \label{prop: infinte-hilbert-capacity}
    Let \(\mathcal{T}: \operatorname{Tr}(H) \to \operatorname{Tr}(H)\) be a  completely positive, trace-preserving (CPTP) map, where \(H\) is an infinite-dimensional separable Hilbert space. Assume 
    \begin{itemize}
        \item $\mc T$ be contractive in  the  operator norm.
        \item \(\Lambda_{\mathcal{T}}\) is a finite-dimensional subspace (spanned by a finite number of generators).
    \end{itemize}
   
   Then, $\Lambda_\mathcal{T}$
   has the form in~\eqref{eq:peripheral-decomposition-1}.

   As a consequence, the classical and quantum capacities of channel $\mc T$ satisfy:
    \begin{equation}    
        \label{eq:classical-cap-formula infinite}
         \log
        \left(\left\lfloor \frac{\sum_k d_k}{1 - \delta} \right\rfloor \right)\leq C_{\delta}^\infty(\mathcal{T}) ,
    \end{equation}
    and 
    \begin{equation}
        \label{eq:quantum-cap-bound infinite}
        \log\left(\left\lfloor \frac{\max_k d_k}{\sqrt{1 - \delta}} \right\rfloor\right) \leq Q_{\delta}^\infty (\mathcal{T}).
    \end{equation}

\end{proposition}

\section{Examples}
\label{sec: examples}
In this section, we present three examples to illustrate our theorem, each highlighting a different type of quantum dynamic. Example~\ref{ex: collective noise} considers $\mathcal{T}$ as noise affecting the system. Example~\ref{ex: cat code} examines an engineered dynamic designed to control the evolution. Finally, Example~\ref{ex: noise+ recovery} explores a scenario where noise, engineered dynamics, and a recovery process interact to restore the system.

\begin{example}\label{ex: collective noise}
    Let $\mathcal{T}$ be a collective noise that simultaneously affects multiple qubits due to uniform interactions with the environment. For an $n$-qubit system, $\mathcal{T}$ is given by the Kraus operators $\{ \frac{1}{\sqrt{3}} \exp(i \sigma^{(n)}_j)\}_{j\in \{x,y,z\}} $, where $$\sigma^{(n)}_j= \sum_{k=1}^n I^{\otimes k-1}_2 \otimes \sigma_j \otimes I^{\otimes n-k}_2,$$ and $\sigma_j$ are Pauli matrices~\cite{holbrook2003noiseless}. 

For a \textbf{4-qubit} system, the only eigenvalue with unit modulus is $1+0i$ (easily checked numerically), 
thus the peripheral subspace of \(\mc T\) is same as its fixed point subspace, and  has the structure \(\mathcal{M}_2 \oplus (\mathcal{M}_3 \otimes I_3) \oplus \mathbb{C} I_5\)~\cite{holbrook2003noiseless}. Therefore, by Theorem~\ref{thm:classical and quantum capacity of repeated channel}  its long-time classical and quantum capacity are 
\(
C_\delta^\infty(\mc T) =\log \left\lfloor \frac{6}{1 - \delta} \right\rfloor\text{, and } \quad  Q_{\delta}^\infty (\mc T) \simeq \log\left(\frac{3}{1 - \delta}\right)
\).

\end{example}

    \begin{example}[Cat code]\label{ex: cat code}
    In quantum information, \textit{cat codes} are an approach to encode qubits in coherent superpositions of photon-number states, known as \textit{cat states}. These states are particularly valuable for creating robust qubits that resist certain types of noise, which is crucial for quantum error correction~\cite{mazyar_lecture,albert2022_lecture}. 

    To construct an \(n\)-component cat code, the system must be engineered so that the system and environment exchange photons while preserving a fixed photon-number parity modulo \(n\)~\cite{mazyar_lecture}. This evolution can be modeled by a QMS with a generator \(\mc{L}\), where the jump operator is given by \(L = a^n - \alpha^n I\), with \(a\) as the photon annihilation operator and \(\alpha \in \mathbb{C}\) as a parameter that controls the steady-state properties.

This dissipative process gradually drives the system toward the \(n\)-component cat code subspace~\cite{albert2022_lecture}. As a result, \textit{cat codes} exhibit resilience against common errors, such as dephasing, allowing them to maintain coherence over extended timescales.

The peripheral subspace of $n$-photon driven  dissipation evolution  
is identical to its fixed-point subspace and is given by  
\[
\chi_{e^{\mc{L}t}} = \operatorname{span} \{ |\beta\rangle \langle \beta'| : \; \beta^n = \beta'^n = \alpha^n \},
\]  
where \(|\beta\rangle\) is a coherent state with parameter \(\beta\)~\cite{albert2022_lecture, mazyar_lecture}. Although the coherent states \(|\beta\rangle\) are not orthogonal to each other, the set of orthogonal states \(\{|\psi_i\rangle\}_{i=0}^{n-1}\), defined as superposition of coherent states as  
\begin{equation} \label{eq: def basis of cat code}
    |\psi_i\rangle = \sum_{j=0}^{n-1} \omega^{ij} |\omega^j |\alpha|\rangle,
\end{equation}  
where \(\omega= \exp(2\pi i/n)\) , satisfies  
\[
\chi_{e^{\mc{L}}} = \operatorname{span}\{ |\psi_i\rangle \langle \psi_j| : i,j = 0, \dots, n-1 \}.
\]  
Thus, the peripheral subspace has the structure \(0 \oplus \mc{M}_n \otimes 1\).

Although Proposition~\ref{prop: infinte-hilbert-capacity} and Theorem~\ref{thm:classical and quantum capacity of repeated channel} are not applicable for
this family of Markov processes—due to the infinite dimensionality of the Hilbert space and the non-contractivity of the processes in the operator norm—the fact that this evolution drives system into $\chi_{\mc L} $~\cite{albert2022_lecture} allows us to leverage the results of Theorem~\ref{thm:classical and quantum capacity of repeated channel}.
Consequently, we have  
\(
C_\delta^\infty(e^{\mc L t}) = \log \left(\left\lfloor \frac{n}{1 - \delta} \right\rfloor \right)\), and \(
 Q_{\delta}^\infty (e^{\mc L t}) \simeq  \log\left(\frac{n}{1 - \delta}\right).
\)
 
\end{example}

\begin{remark}
Although the Hilbert space of the cat code in Example~\ref{ex: cat code} is infinite-dimensional, the coefficients of coherent states in the Fock basis decay exponentially. Thus, we can apply the algorithm in Theorem~\ref{thm:algorithm-structure} by truncating the number of photons, achieving a good accuracy (see~\cite{mostafa_taheri_2024_12821315}).
\end{remark}

\begin{example}\label{ex: noise+ recovery}
    
Photon loss, a common noise in bosonic systems,  is modeled as a Lindblad process with the jump operator \(a\) (the photon annihilation operator)~\cite{albert2022_lecture, mazyar_lecture}. Photon loss disrupts the structure of the cat code by altering the photon-number parity, which is a key property for the code's error-correcting capabilities.

To approximate this noise over a finite time interval, we consider a simplified model where, with probability \(p\), no photon is lost, and with probability \(1 - p\), exactly one photon is lost. The noise channel \(\mathcal{N}\) is described by the Kraus operators:  
\( N_1 = \sqrt{p} I \),  
\( N_2 = \sqrt{1 - p} \sum_{m=1}^{\infty} |m - 1\rangle \langle m| \), and  
\( N_3 = \sqrt{1 - p} |0\rangle \langle 0| \).

A recovery channel \(\mathcal{R}\) can be defined to counteract the effects of photon loss. The Kraus operators for \(\mathcal{R}\), in a simplified version of those given in~\cite{Mazyar_Autonomous_Protection}, are:  
\( R_1 = \sum_{m=0}^\infty |2m+1\rangle \langle 2m| \) and  
\( R_2 = \sum_{m=0}^\infty |2m+1\rangle \langle 2m+1| \).

Let the noise channel \(\mathcal{N}\) and the recovery channel \(\mathcal{R}\) act on the system at each time interval \(\tilde{t}\). Then, the evolution of the system in each time interval \(\tilde{t}\) is given by:  
\[
\mathcal{T} = \mathcal{R} \circ \mathcal{N} \circ e^{\mathcal{L} \tilde t},
\]
where \(\mathcal{L}\) is the generator of the QMS describing the system's dynamics.

Let us focus on the 4-component cat code setup (\(n=4\)). One can verify that  
\[
\chi_\mc{T} = \operatorname{span}\{ |\psi^i\rangle\langle\psi^j| : i,j \in \{1,3\} \},
\]  
where \(|\psi^i\rangle\) are defined as in \eqref{eq: def basis of cat code}.  
As a result, instead of having zero capacity in the long-time limit due to noise, the system achieves a classical capacity of  
\[
C_\delta^\infty(\mathcal{T}) = \log \left\lfloor \frac{2}{1 - \delta} \right\rfloor,
\]  
and a quantum capacity satisfying  
\[
Q_\delta^\infty(\mathcal{T}) \geq \log \left\lfloor \frac{2}{\sqrt{1 - \delta}} \right\rfloor.
\]

    \end{example}

 \section{Conclusion }
In summary, we have shown that in the setting of infinite time, the classical and quantum capacities are essentially given by the zero-error capacities and can be efficiently computed, unlike the fixed time setting. We see this result as a first step towards understanding channel capacities for quantum evolutions. It would be interesting to study the behavior of capacities for finite $t$ both for active and passive error correction.

 \section*{Acknowledgements}
    We acknowledge support from the European Research Council (ERC Grant AlgoQIP, Agreement No. 851716). MR is supported by the Marie Sk\l odowska- Curie Fellowship from the European Union’s Horizon Research and Innovation programme, grant Agreement No. HORIZON-MSCA-2022-PF-01 (Project number: 101108117).
    We would like to thank Nilanjana Datta for a talk at the ``Quantum information" conference at the SwissMAP Research Station in Les Diablerets on~\cite{singh2024-v1} which stimulated this work.

 \bibliographystyle{plain}

    \bibliography{reference}

\begin{thebibliography}{10}

\bibitem{albert2022_lecture}
Victor~V. Albert.
\newblock Bosonic coding: introduction and use cases.
\newblock {\em arXiv preprint arXiv:2211.05714}, 2022.
\newblock Available at \url{https://arxiv.org/abs/2211.05714}.

\bibitem{bapat1997nonnegative}
Ravi~B Bapat and Tirukkannamangai~ES Raghavan.
\newblock {\em Nonnegative matrices and applications}.
\newblock Number~64. Cambridge university press, 1997.

\bibitem{bardet2021group}
Ivan Bardet, Marius Junge, Nicholas Laracuente, Cambyse Rouz{\'e}, and Daniel~Stilck Fran{\c{c}}a.
\newblock Group transference techniques for the estimation of the decoherence times and capacities of quantum markov semigroups.
\newblock {\em IEEE Transactions on Information Theory}, 67(5):2878--2909, 2021.

\bibitem{barman2017algorithmic}
Siddharth Barman and Omar Fawzi.
\newblock Algorithmic aspects of optimal channel coding.
\newblock {\em IEEE Transactions on Information Theory}, 64(2):1038--1045, 2017.

\bibitem{blume2010information}
Robin Blume-Kohout, Hui~Khoon Ng, David Poulin, and Lorenza Viola.
\newblock Information-preserving structures: A general framework for quantum zero-error information.
\newblock {\em Phys. Rev. A}, 82:062306, Dec 2010.

\bibitem{brown2016quantum}
Benjamin~J. Brown, Daniel Loss, Jiannis~K. Pachos, Chris~N. Self, and James~R. Wootton.
\newblock Quantum memories at finite temperature.
\newblock {\em Rev. Mod. Phys.}, 88:045005, Nov 2016.

\bibitem{denardo1977periods}
Eric~V Denardo.
\newblock Periods of connected networks and powers of nonnegative matrices.
\newblock {\em Mathematics of Operations Research}, 2(1):20--24, 1977.

\bibitem{fawzi2022lower}
Omar Fawzi, Alexander M\"{u}ller-Hermes, and Ala Shayeghi.
\newblock {A Lower Bound on the Space Overhead of Fault-Tolerant Quantum Computation}.
\newblock In Mark Braverman, editor, {\em 13th Innovations in Theoretical Computer Science Conference (ITCS 2022)}, volume 215 of {\em Leibniz International Proceedings in Informatics (LIPIcs)}, pages 68:1--68:20, Dagstuhl, Germany, 2022. Schloss Dagstuhl -- Leibniz-Zentrum f{\"u}r Informatik.

\bibitem{GKS}
Vittorio Gorini, Andrzej Kossakowski, and E.~C.~G. Sudarshan.
\newblock {Completely positive dynamical semigroups of N‐level systems}.
\newblock {\em Journal of Mathematical Physics}, 17(5):821--825, 05 1976.

\bibitem{gottesman2016surviving}
Daniel Gottesman.
\newblock Surviving as a quantum computer in a classical world.
\newblock {\em Textbook manuscript preprint}, 2016.

\bibitem{GUAN201855}
Ji~Guan, Yuan Feng, and Mingsheng Ying.
\newblock Decomposition of quantum markov chains and its applications.
\newblock {\em Journal of Computer and System Sciences}, 95:55--68, 2018.

\bibitem{mazyar_lecture}
Jérémie Guillaud, Joachim Cohen, and \mbox{Mazyar} Mirrahimi.
\newblock {Quantum computation with cat qubits}.
\newblock {\em SciPost Phys. Lect. Notes}, page~72, 2023.

\bibitem{holbrook2003noiseless}
John~A Holbrook, David~W Kribs, and \mbox{Raymond} Laflamme.
\newblock Noiseless subsystems and the structure of the commutant in quantum error correction.
\newblock {\em Quantum Information Processing}, 2:381--419, 2003.

\bibitem{horn2013matrix}
R.~A. Horn and C.~R. Johnson.
\newblock {\em Matrix Analysis}.
\newblock Cambridge University Press, 2013.

\bibitem{Mazyar_Autonomous_Protection}
Zaki Leghtas, Gerhard Kirchmair, Brian \mbox{Vlastakis}, Robert~J. Schoelkopf, Michel~H. Devoret, and \mbox{Mazyar} Mirrahimi.
\newblock Hardware-efficient autonomous quantum memory protection.
\newblock {\em Phys. Rev. Lett.}, 111:120501, Sep 2013.

\bibitem{lindblad1976generators}
Goran Lindblad.
\newblock On the generators of quantum dynamical semigroups.
\newblock {\em Communications in mathematical physics}, 48:119--130, 1976.

\bibitem{muller2014quantum}
Alexander M{\"u}ller-Hermes, David Reeb, and Michael~M Wolf.
\newblock Quantum subdivision capacities and continuous-time quantum coding.
\newblock {\em IEEE Transactions on Information Theory}, 61(1):565--581, 2014.

\bibitem{olkiewicz_1999}
R.~Olkiewicz.
\newblock Environment-{Induced} {Superselection} {Rules} in {Markovian} {Regime}.
\newblock {\em Communications in Mathematical Physics}, 208(1):245--265, December 1999.

\bibitem{pirandola2019end}
Stefano Pirandola.
\newblock End-to-end capacities of a quantum communication network.
\newblock {\em Communications Physics}, 2(1):51, 2019.

\bibitem{rivas2012open}
Angel Rivas and Susana~F Huelga.
\newblock {\em Open quantum systems}, volume~10.
\newblock Springer, 2012.

\bibitem{shannon1956zero}
Claude Shannon.
\newblock The zero error capacity of a noisy channel.
\newblock {\em IRE Transactions on Information Theory}, 2(3):8--19, 1956.

\bibitem{singh2024-v1}
Satvik Singh, Mizanur Rahaman, and \mbox{Nilanjana} Datta.
\newblock Zero-error communication, scrambling, and ergodicity.
\newblock {\em arXiv preprint arXiv:2402.18703v1}, 2024.
\newblock Available at \url{https://arxiv.org/abs/2402.18703v1}.

\bibitem{singh2024-v2}
Satvik Singh, Mizanur Rahaman, and \mbox{Nilanjana} Datta.
\newblock Zero-error communication under discrete-time markovian dynamics.
\newblock {\em arXiv preprint arXiv:2402.18703v2}, 2024.
\newblock Available at \url{https://arxiv.org/abs/2402.18703v2}.

\bibitem{mostafa_taheri_2024_12821315}
Mostafa Taheri.
\newblock Code for computing the peripheral structure of a quantum channel.
\newblock Available at \url{https://github.com/mostafataheri75}/structure-of-the-peripheral-subspace.

\bibitem{tarjan1972depth}
Robert Tarjan.
\newblock Depth-first search and linear graph algorithms.
\newblock {\em SIAM journal on computing}, 1(2):146--160, 1972.

\bibitem{terhal2015quantum}
Barbara~M. Terhal.
\newblock Quantum error correction for quantum memories.
\newblock {\em Rev. Mod. Phys.}, 87:307--346, Apr 2015.

\bibitem{Watrous_2018}
John Watrous.
\newblock {\em The Theory of Quantum Information}.
\newblock Cambridge University Press, 2018.

\bibitem{Wolfe2012}
Michael~M. Wolf.
\newblock Quantum channels and operations - guided tour, 2012.
\newblock Graue Literatur.

\end{thebibliography}

\newpage

    \appendix

    \section{Infinite-time capacities }
    As previously mentioned, the peripheral subspace of the channel plays a crucial role in determining the infinite-time capacity. Before diving into the proof, let us first highlight some of its key properties.
        \begin{proposition}\cite{Wolfe2012}
    \label{prop:peripheral-space}
    The peripheral subspace of a quantum channel $\mc{T}$ is defined by
    \[
    \chi_{\mc{T}} = \mathrm{span}\{ X \in \mc{M}_d : \mc{T}(X) = \lambda X, \text{ s.t. } |\lambda| = 1\}.
    \]
    It satisfies the following:
    \begin{itemize}
    \item There exists a direct sum decomposition of the Hilbert space $\mb{C}^d = H_0 \oplus \bigoplus_{k=1}^K (H_{k,1} \otimes H_{k,2})$ for some nonnegative integer $K$ such that 
    \begin{align}
    \label{eq:peripheral-decomposition}
        \chi_{\mc{T}} = 0 \oplus \bigoplus_{k=1}^K \mc{M}_{d_k} \otimes \omega_k,
    \end{align}
    where $\mc{M}_{d_k}$ is the full matrix algebra on $H_{k,1}$, with $d_k = \dim H_{k,1}$, and $\omega_k$ is a density operator on $H_{k,2}$.
    \item There exists unitaries $U_k$ acting on $H_{k,1}$ and a permutation $\pi$ which permutes systems $\{1, \dots, K\}$ having the same dimension such that for every $X \in \chi_{\mc{T}}$ of the form $X = 0 \oplus \bigoplus_{k=1}^K x_k \otimes \omega_k$, we have
    \begin{align*}
    \mc{T}(X) = 0 \oplus \bigoplus_{k=1}^K U_k x_{\pi^{-1}(k)} U_{k}^{\dagger} \otimes \omega_k.
    \end{align*}
    \end{itemize}
    \end{proposition}

\subsection{Proof of Theorem~\ref{thm:classical and quantum capacity of repeated channel}}
    \begin{proof}
    We start with the achievability statement, i.e., lower bounds on the capacities. Using the decomposition of the space in Proposition~\ref{prop:peripheral-space}, $\mb{C}^d = H_0 \oplus \bigoplus_{k=1}^K (H_{k,1} \otimes H_{k,2})$, let $\{\ket{e_{k,j}}\}_{j=1}^{d_k}$ be orthonormal bases of $H_{k,1}$. For different $H_{k,1}$ having the same dimension, we will identify the orthonormal bases.

    We define the classical code $\mc{E}$ as follows.  Let $\sigma : \left\{1, \dots, \sum_{k=1}^K d_k\right\} \to \{(k,j) : k \in \{1, \dots, K\}, j \in \{1, \dots, d_k\}\}$ be an arbitrary bijection, then 
    \begin{align*}
    \mc{E}(\proj{i}) = 
    \left\{\begin{array}{cc} \proj{e_{\sigma(i)}} \otimes \omega_k & \text{if } 1 \leq i \leq \sum_{k=1}^K d_k \\
    \proj{e_{1}} \otimes \omega_k & \text{if } \sum_{k=1}^K d_k < i \leq D.
    \end{array} \right.
    \end{align*}
    For $1 \leq i \leq \sum_{k=1}^{K} d_k$, we have $\cE(\proj{i}) = \proj{e_{k,j}} \otimes \omega_k \in B(H_{k,1} \otimes H_{k,2})$ with $(k,j) = \sigma(i)$. As a result, applying $\cT$, we get
    \begin{align*}
    \cT(\cE(\proj{i})) &= U_{\pi(k)} \proj{e_{k,j}} \otimes \omega_{\pi(k)} U_{\pi(k)}^\dagger\\
    &\quad  \in B(H_{\pi(k),1} \otimes H_{\pi(k),2}),
\end{align*}

    using Proposition~\ref{prop:peripheral-space}.
    Composing $t$ times the map $\cT$, we get

   \begin{align*}
    \cT^t&(\cE(\proj{i})) = \\
    &U_{\pi^t(k)} \cdots U_{\pi(k)} \proj{e_{k,j}} 
    U^{\dagger}_{\pi(k)} \cdots U^{\dagger}_{\pi^t(k)} 
     \otimes \omega_{\pi^t(k)} \\
    & \in B(H_{\pi^t(k),1} \otimes H_{\pi^t(k),2}).
\end{align*}

    Note that we always have $\cT^t(\cE(\proj{i})) \in \chi_{\cT}$. For a time $t$, we choose a recovery map $\cR$ defined by

    \begin{align*}
    \cR(X) = \sum_{k=1}^K &\bigg[ V_k U^{\dagger}_{\pi(k)} \cdots U^{\dagger}_{\pi^t(k)} \\
    &\quad\tr_{H_{\pi^t(k),2}}(P_{\pi^t(k)} X P_{\pi^t(k)}) \\
    &\quad \quad U_{\pi^t(k)} \cdots U_{\pi(k)} V_{k}^{\dagger}\bigg],
    \end{align*}
    where $P_{k}$ is the orthogonal projection onto $H_{k,1} \otimes H_{k,2}$ and $V_{k} = \sum_{j=1}^{d_k} \ketbra{\sigma^{-1}(k,j)}{e_{k,j}}$. Clearly, an error occurs only if $i > \sum_{i=1}^K d_K$ and as such the success probability is given by
    \begin{align*}
    \frac{1}{D} \sum_{i=1}^D \bra{i} \cR \circ \cT^t \circ \cE( \proj{i} ) \ket{i} = \min(1, \frac{\sum_{k=1}^D d_k }{D}).
    \end{align*}
    As a result, if we choose $D = \floor{\frac{\sum_{k=1}^{K} d_k}{1-\delta}}$, we obtain the desired achievability.

    Now let us construct a quantum code of dimension $D$ satisfying the condition $\max_{(\gamma_k)_k} \{ \frac{\sum_{k=1}^{K} \gamma_k^2}{D^2} \} \geq 1-\delta$.
   
    Let $\gamma_k$ achieve the maximum. It is convenient to label the basis of $\mb{C}^D$ by elements in $S \cup S'$, where $S = \{(k,j) : k \in \{1, \dots, K\}, j = 1 \in \{1, \dots, \gamma_k\}\}$ and $S' = \{1, \dots, D-\sum_{k=1}^{K} \gamma_k\}$. By the condition $\sum_{k = 1}^K \gamma_k \leq D$, the set $S$ is of size at most $D$.
    
    In order to define our encoding map $\cE$, we first define the Kraus operators $E_{k} = \sum_{j=1}^{\gamma_k} \ketbra{e_{k,j}}{(k,j)}$ for $k = 1$ to $K$, and $F_i = \ketbra{e_{(1,1)}}{i}$ for $i \in S'$.
    Then for $X \in \cM_{D}$ we define $\cE(X) = \sum_{k=1}^{K} (E_k X E_k^{\dagger}) \otimes \omega_k + \sum_{i \in S'} F_i X F_i^{\dagger} \otimes \omega_1$. As a result, we have
    \[
     \cE(\ketbra{(k,j)}{(k,j')}) = \ketbra{e_{k,j}}{e_{k,j'}} \otimes \omega_k,
    \]
    and $\cE(\ketbra{(k,j)}{(k',j')}) = 0$ for $k \neq k'$. For $i, i' \in S'$ we get
    \[
     \cE(\ketbra{(k,j)}{i}) = 0,
    \]
    \[
     \cE(\ketbra{i}{i}) = \ketbra{e_{1,1}}{e_{1,1}} \otimes \omega_1 \quad \text{, and }\]\[ \quad \cE(\ketbra{i}{i'}) = 0 \text{ for } i\neq i'.
    \]        
    Thus, applying $\cT^t$, we get

    \begin{align*}
        (\cT^t \circ \cE)&(\ketbra{(k,j)}{(k,j')}) =\\
       & U_{\pi^t(k)} \cdots U_{\pi(k)} \ketbra{e_{k,j}}{e_{k,j'}} U_{\pi(k)}^{\dagger} \cdots U_{\pi^t(k)}^{\dagger}\\
       & \quad \otimes \omega_{\pi^t(k)}.
    \end{align*}
    Note that $\cT^t \circ \cE(\ketbra{(k,j)}{(k',j')}) \in \chi_{\cT}$. We choose $\cR$ as

    \begin{align*}
        \cR(X) = \sum_{k=1}^K& \bigg[V_k  U^{\dagger}_{\pi(k)} \cdots U^{\dagger}_{\pi^t(k)} \\
        &\quad \tr_{H_{\pi^t(k),2}}(P_{\pi^t(k)} X P_{\pi^t(k)})\\
        & \quad \quad U_{\pi^t(k)} \cdots U_{\pi(k)} V_k^{\dagger}\bigg],
    \end{align*}
    where $P_k$ is the orthogonal projection onto $H_{k,1} \otimes H_{k,2}$ and $V_k = \sum_{j=1}^{\gamma_k} \ketbra{(k,j)}{e_{k,j}}$.
    
    Then we can compute the entanglement fidelity as
    \begin{align*}
    &\bra{\Phi^+} \left(\mc{I} \otimes \mc R\circ \mc T^t \circ \mc E \right)\left(|\Phi^+\ra\la \Phi^+|\right) \ket{\Phi^+} \\
    &\geq \frac{1}{D^2} \sum_{(k,j), (k',j') \in S}\bigg[\\ 
    &\;\;\;\;\bra{(k,j)} \cR \circ \cT^t \circ \cE(
   \ketbra{(k,j)}{(k',j')} ) \ket{(k',j')} \Bigg]\\
    &= \frac{1}{D^2} \sum_{k=1}^{K} \sum_{j,j'=1}^{\gamma_k} 1 \\    
    &= \frac{\sum_{k=1}^K \gamma_k^2 }{D^2} \\
    &\geq 1-\gamma, 
    \end{align*}
    which proves the desired result.

    We now move to the converse bound. For this, it is convenient to consider the peripheral projection channel $\mc{T}_P$~\cite[Proposition 6.3]{Wolfe2012} which satisfies the following properties:~\cite{Wolfe2012}
        \begin{itemize}
            \item $\mc T_P(\mc M_d)= \chi_{\mc T}$
            \item there exist an increasing sequence $\{t_i\}$ such that $\lim_{i\to \infty} \mc T^{t_i} = \mc T_P$
            \item $\mc T_P(X) = X$ for any $X \in \chi_\mc T$.
        \end{itemize}

    Note that if $\cE$ is a $(D, \delta)$ code for $\cT$, and $\|\cT - \cT'\|_{\diamond} \leq \eta$, where the diamond norm is defined as $\| \cT - \cT' \|_{\diamond} = \sup_{\rho} \|(\mc{I} \otimes (\cT - \cT'))(\rho) \|_1$, then $\cE$ is also a $(D, \delta + \eta)$ code for $\cT'$. 
    
    Consider a sequence $t_i$ such that $\lim_{i \to \infty} \cT^{t_i} = \cT_P$ and let $\eta_i = \| \cT^{t_i} - \cT_P \|_{\diamond}$. Let $\cE_i$ be a $(D, \delta)$ code for the channel $\cT^{t_i}$. Then $\cE_i$ is also a $(D, \delta + \eta_i)$ code for the channel $\cT_{P}$. Thus $C_{\delta}(\cT^{t_i}) \leq C_{\delta + \eta_i}(\cT_{P})$.
    Taking the limit $i \to \infty$, we have that $\sup_{\delta' < \delta} C_{\delta'}(\cT_P)\leq C^{\infty}_{\delta}(\cT) \leq \inf_{\delta' > \delta} C_{\delta'}(\cT_P)$. The exact same result holds for the quantum capacity as well. It now suffices to find upper bounds on the capacities of the channel $\cT_P$.

    First, let us show that the classical and quantum capacities of $\cT_P$ and $\sum_{k=0}^K \cA_k \circ \cT_P$ are the same, where we define $\cA_k$ as follows.
    Define $P_0$ to be the orthogonal projector onto $H_0$ and $\cA_0(X) = P_0 X P_0$. Then define $P_k$ to be the orthogonal projector onto $H_{k,1} \otimes H_{k,2}$ and $\cA_k : B(\mb{C}^d) \to B(H_{k,1})$ by $\cA_k(X) = \tr_{H_{k,2}}(P_k X P_k)$. Then $\sum_{k=0}^K \cA_k \circ \cT_P$ is clearly a quantum channel. As $\cT_P(\cM_d) = \chi_{\cT}$, $\sum_{k=0}^K \cA_k \circ \cT_P = \sum_{k=1}^K \cA_k \circ \cT_P$ so  $\sum_{k=1}^K \cA_{k} \circ \cT_P$ is a quantum channel mapping $B(\mb{C}^d)$ to $\bigoplus_{k=1}^K B(H_{k,1})$.  
    The inequality $C_{\delta}(\sum_{k=1}^K \cA_k \circ \cT_P) \leq C_{\delta}(\cT_P)$ is clear. For the other inequality, let $\cE$ be a code for $\cT_P$ and $\cR$ be a corresponding recovery map. We define the recovery map $\cR' = \sum_{k=1}^K \cR \circ \cB_k$ where $\cB_k : B(H_{k,1}) \to B(H_{k,1} \otimes H_{k,2})$ and maps $\cB_k(x_k) = x_k \otimes \omega_k$. It is easy to see that
    \[
    \cR' \circ \sum_{k=1}^K \cA_k \circ \cT_P \circ \cE = \cR \circ \cT_P \circ \cE,
    \]
    which proves the desired statement. For this reason, in what follows, we assume that $\dim H_0 = 0$ and $\dim H_{k,2} = 1$ and the space decomposes as $\mb{C}^d = \bigoplus_{k=1}^K H_{k,1}$.

    Let $\mc E$ be a $(D,\delta)$ classical code for $\mc T_P$ and $\mc R$ a corresponding recovery channel. Note that because $\cT_P(\mc{M}_d) = \chi_{\cT}$, for any density operator $\rho \in B(\mb{C}^d)$, there exists positive operators $\rho_k$ on $H_{k,1}$, $\cT_P(\rho) = \bigoplus_{k=1}^K \rho_k$. In addition, as $\cT_P$ is trace-preserving $\sum_{k} \tr(\rho_k) = 1$ and $\rho_k \leq I_{H_{k,1}}$, where $I_{H_{k,1}}$ is the identity on $H_{k,1}$. As a result, $\cT_P(\rho) \leq  \bigoplus_{k=1}^K I_{H_{k,1}}$. Thus,
    \begin{align*}
            \frac 1D\sum_{i=1}^{D} \bra{i} \mc R \circ \cT_P \circ \mc E(\proj{i}) \ket{i} 
            &\leq \frac 1D \sum_{i=1}^D \bra{i} \cR \left( \bigoplus_{k=1}^K I_{H_{k,1}} \right) \ket{i} \\
             &= \frac 1D \tr\left( \bigoplus_{k=1}^K I_{H_{k,1}} \right)\\
            & = \frac{\sum_{k=1}^K d_k}{D}.
    \end{align*}
    Thus, for any $(D, \delta)$ code, we should have $\frac{\sum_{k=1}^K d_k}{D} \geq 1-\delta$ which implies $D \leq \floor{\frac{\sum_{k=1}^K d_k}{1-\delta}}$. This implies that 
    $C_{\delta}(\cT_P) \leq \floor{\frac{\sum_{k=1}^K d_k}{1-\delta}}$. As a result $C^{\infty}_{\delta}(\cT) \leq \floor{\frac{\sum_{k=1}^K d_k}{1-\delta'}}$ for any $\delta' > \delta$ which gives the desired result.

    Let us now move to the quantum capacity. Let $\mc E$ be a $(D,\delta)$ classical code for $\mc T_P$ and $\mc R$ a corresponding recovery channel. 
    
    In order to compute the entanglement fidelity, recall that the  entanglement fidelity of channel $\mc E$ with Kraus operators $\{E_i\}$ is given by $F_E(\mc E) = \sum_i \left|\tr(E_i)\right|^2$. Let us introduce Kraus operators for $\{E_j\}_j$ for $\cT_P \circ \cE$, $\{R_{i}\}_i$ for $\cR$. As $\cT_P(\cM_d) = \chi_{\cT}$, the operator $\{P_k E_j\}_{k, j}$ are also Kraus operators for the map $\cT_P \circ \cE$, where we recall that $P_k$ is the projector onto $H_{k,1}$.

    Then we have
            \begin{align*}
                F_E(\mc R \circ \cT_P \circ \mc E) &= \frac {1}{D^2}\sum_{i,j,k}\left|\tr(R_{i} P_k E_{j})\right|^2.
            \end{align*}
    Let us denote $\alpha_k^2 = \sum_{i,j} \left|\tr(R_{i}P_k E_{j})\right|^2$. We show two properties on $\alpha_k$: $\alpha_k^2 \leq d_k \beta_k$ with $\sum_{k} \beta_k = D$ and $\alpha_k^2 \leq d_k^2$. We start with the first property
            \begin{align*}
                \alpha_k^2 &\leq \sum_{i,j} \tr( R_{i} P_k P_k^{\dagger} R_{i}^{\dagger} ) \tr(E_{j}^{\dagger} P_k^{\dagger} P_k E_{j})\\
                &= \tr(P_k) \sum_{j} \tr\left( E_{j}^{\dagger} P_k E_j  \right).
            \end{align*}    
For the first inequality, we used the Cauchy-Schwarz inequality $|\tr(A^{\dagger} B)|^2 \leq \tr(A^{\dagger} A) \tr(B^{\dagger} B)$. Then we used the fact that $\cR$ is trace preserving. Now note that because $\cT_P \circ \cE$ is a quantum channel, we have $\sum_{k,j} \tr\left(  E_{j}^{\dagger} P_{k}^{\dagger} P_k E_j \right) = D$. Calling $\beta_k = \sum_{j} \tr\left( E_{j}^{\dagger} P_{k}^{\dagger} P_k E_j \right)$, we proved the first claimed inequality. 

To show the second inequality $\alpha_k^2 \leq d_k^2$, we write
            \begin{align*}
                \alpha_k^2 &\leq \sum_{i,j} \tr( E_j R_{i} P_k P_k^{\dagger} R_{i}^{\dagger} E_j^{\dagger} ) \tr( P_k^{\dagger} P_k )\\
                &= \tr(P_k) d_k \\
                &= d_k^2 .
            \end{align*} 
where we used again the Cauchy-Schwartz inequality and the fact that $\cT_P \circ \cE \circ \cR$ is trace-preserving. Defining $\gamma_k = \min(d_k, \beta_k)$, we have that $\sum_{k=1}^K \gamma_k \leq D$ and $\gamma_k \leq d_k$ and the entanglement fidelity can be bounded as 
\begin{align*}
                F_E(\mc R \circ \cT_P \circ \mc E) 
                &\leq \frac {1}{D^2}\sum_{k=1}^K \alpha_k^2 \\
                &\leq \frac{1}{D^2}\sum_{k=1}^K d_k \gamma_k ,
\end{align*}
which concludes the proof of~\eqref{eq:quantum-cap-formula}.
\end{proof}

 \begin{remark}
    For the quantum capacity, we do not have an exact expression for the capacity but upper and lower bounds that differ by roughly $\frac{1}{2} \log(1-\delta)$. The upper and lower bounds we establish in the proof is slightly stronger and it is simpler to express in terms of the optimal error $\delta$ for a fixed code size $D$:
    \begin{equation}
        \label{eq:quantum-cap-formula}
        \begin{aligned}
            &\max_{(\gamma_k)_k} \frac{\sum_{k=1}^{K} \gamma_k^2}{D^2} \\
            &\leq\sup \{1- \delta : \exists (D, \delta) \text{ quantum code for } \cT \} \\
            &\leq \max_{(\gamma_k)_k}\frac{\sum_{k=1}^K d_k \gamma_{k}}{D^2}
        \end{aligned}
    \end{equation}
    where the supremum is taken over integers $\gamma_k$ satisfying $\gamma_k \leq d_k$ and $\sum_{k=1}^K \gamma_k \leq D$. Note that it is simple to see that the optimal choice for $\gamma_k$ is simply to take $\gamma_1 = d_1, \dots, \gamma_s = d_s$ and $\gamma_{s+1} = D - \sum_{k=1}^s \gamma_k$ where $s = \argmax\{ s : \sum_{k=1}^s d_k < D \}$. When $D$ is the sum of the largest $r$ elements of $(d_k)_k$, then the upper and lower bounds match and we obtain an exact characterization. Observe that~\eqref{eq:quantum-cap-bound} follows from~\eqref{eq:quantum-cap-formula} by observing that $\max_{(\gamma_k)_k} \sum_{k=1}^{K} \gamma_k^2 \geq \min\left(D^2, \max_k d_k^2\right)$ and 
    $\sum_{k=1}^K d_k \gamma_{k} \leq (\max_{k} d_k) D$. We leave the problem of computing the exact optimal fidelity as an open problem.
    \end{remark}

        \begin{remark}[Classical special case]
    Note that a classical stochastic $d \times d$ matrix $M$ can be seen as a special case of a quantum channel $\cT( \ketbra{i}{i'}) = 0$ if $i \neq i'$ and $\cT( \ketbra{i}{i}) = \sum_{j} M_{j,i} \proj{j}$. It is straightforward to verify that the eigenvalues of $M$ are the eigenvalues of $\cT$ and that the quantity $\sum_{k=1}^K d_k$ for such a channel is the number of eigenvalues of $M$ of modulus $1$ counted with multiplicity or the dimension of the peripheral subspace. This quantity even has a combinatorial interpretation as the sum of the periods of the bottom strongly connected components of the directed graph associated with the Markov chain $M$, as shown in~\cite{GUAN201855}. Note that this combinatorial interpretation holds more generally for nonnegative matrices and such a decomposition into bottom strongly connected components is sometimes called the Frobenius normal form of the matrix~\cite[Section 1.7]{bapat1997nonnegative}, see also e.g., in~\cite[Theorem 8.5.3 and Remark 8.5.4]{horn2013matrix} for the period of each component. Using the combinatorial interpretation, one can find an algorithm running in linear time in the size of the graph for computing the capacity by using e.g., Tarjan's algorithm~\cite{tarjan1972depth} to find the bottom strongly connected components and then find the periods of each component using~\cite{denardo1977periods}.
 
    \end{remark}

  \subsection{Quantum Markov semi-group's infinite-time capacities}
  \begin{proposition} \label{prop: Markov semigroup capacities}
        Let $(\cT_t)_{t \geq 0}$ be a quantum Markov semigroup. Then we have
        \begin{align*}
        \lim_{t \to \infty} C_{\delta}(\cT_{t}) = C^{\infty}_{\delta}(\cT_1) \\
        \lim_{t \to \infty} Q_{\delta}(\cT_{t}) = Q^{\infty}_{\delta}(\cT_1).
        \end{align*}
        In addition, the peripheral subspace $\chi_{\cT_1}$ can be expressed in terms of the spectrum of generator $\mc{L}$ as follows:
             \[ \chi_{\cT_1} = \text{span}\{ X \in \mc{M}_d | \exists \theta\in \mb R ,\;\; \mc L(X) = i\theta X \}.\] 
    \end{proposition}
\begin{remark}
        In Proposition~\ref{prop: Markov semigroup capacities}, the choice of \(t = 1\) in \(\cT_1\) is arbitrary and made for simplicity. The result holds for any fixed \(t_0 > 0\), with \(\cT_1\) replaced by \(\cT_{t_0}\).
    \end{remark} 

    We note that the behavior capacity of a special family of QMS (transferred QMS) over time has been thoroughly analyzed in~\cite{bardet2021group}, providing an asymptotic capacity evolution over time. 
\begin{proof}
            Let \(\mathcal{E}\) be a \((D,\delta)\) classical (or quantum) code for \(\mathcal{T}_t\), and \(\mathcal{R}\) be its corresponding recovery channel. If \(t' \leq t\), then \(\mathcal{E}\) is a \((D,\delta)\) classical (or quantum) code for \(\mathcal{T}_{t'}\) with recovery operator \(\mathcal{R} \circ \mathcal{T}_{t-t'}\). Therefore, we have the following inequalities:
            \[
            C_\delta(\mathcal{T}_{\lfloor t \rfloor}) \geq C_\delta(\mathcal{T}_t) \geq C_\delta(\mathcal{T}_{\lceil t \rceil})
            \]
            and
            \[
            Q_\delta(\mathcal{T}_{\lfloor t \rfloor}) \geq Q_\delta(\mathcal{T}_t) \geq Q_\delta(\mathcal{T}_{\lceil t \rceil}).
            \]
            In the limit as \(t \to \infty\), both \(C_\delta(\mathcal{T}_{\lfloor t \rfloor})\) and \(C_\delta(\mathcal{T}_{\lceil t \rceil})\) converge to \(C_\delta^\infty(\mathcal{T}_1)\). Therefore, we have
            \[
            \lim_{t \to \infty} C_\delta(\mathcal{T}_t) = C_\delta^\infty(\mathcal{T}_1)
            \quad \text{and} \quad \lim_{t \to \infty} Q_\delta(\mathcal{T}_t) = Q_\delta^\infty(\mathcal{T}_1).
            \]
        
            Next, we examine the peripheral subspace of  \(\mathcal{T}_1\). It is straightforward to see that if \(\mathcal{L}(X) = i\theta X\) for some operator \(X\) and real number \(\theta\), then  \(\mathcal{T}_t(X) = e^{t\mathcal{L}}(X) = e^{i\theta t}X\), which implies \(X \in \chi_{\mc T_t}\). Consequently, if \(X\) is spanned by eigenoperators of \(\mathcal{L}\) corresponding to imaginary eigenvalues, then \(X\) belongs to \(\chi_{\mathcal{T}_1}\).

      To show the converse, we use the ``super-operator form" which represents the linear map \(\mathcal{E}\) as a matrix \(\hat{E}\) acting on the vector space of operators. Specifically, for a quantum channel \(\mathcal{E}\) and a density matrix \(\rho\), the action of \(\mathcal{E}\) can be written as \(\text{vec}(\mathcal{E}(\rho)) = \hat{E} \text{vec}(\rho)\), where \(\text{vec}(\cdot)\) denotes the column-stacking vectorization of a matrix.
            Let  \(L\) be the super-operator form of \(\mathcal{L}\), with a Jordan decomposition given by
            \[
            L = V \left( \bigoplus_{\ell=1}^{s} J_{\ell}(\lambda_{\ell}) \right) V^{-1},
            \]
            where \(J_{\ell}(\lambda_{\ell})\) is a Jordan block corresponding to an eigenvalue \(\lambda_{\ell}\) of \(L\), i.e., $J_{\ell}(\lambda_{\ell}) = \lambda_{\ell} I + J_{\ell}(0)$ with $J_{\ell}(0) = \begin{pmatrix} 0 & 1 & 0 & \dots & 0 \\ 0 & 0 & 1 & \dots & 0 \\
            \vdots & \vdots & \ddots & \vdots & \vdots \\ 0 & \dots & \dots & \dots & 0\end{pmatrix} \in \cM_{d_{\ell}}$, where $d_{\ell}$ is the size of the block.
            Thus the super-operator of \(\mathcal{T}_1 = e^{\mathcal{L}}\) has the form
            \begin{equation} \label{eq: jordan exp lt}
                 e^{L} = V \left( \bigoplus_{\ell=1}^s e^{J_\ell(\lambda_\ell)} \right) V^{-1},
            \end{equation}
           with 
\[
e^{J_\ell(\lambda_\ell) }= e^{\lambda_\ell } \cdot \begin{pmatrix}
1 & 1 & \frac{1}{2!} & \cdots & \frac{1}{(d_{\ell}-1)!} \\
0 & 1 & 1 & \cdots & \frac{1}{(d_{\ell}-2)!} \\
0 & 0 & 1 & \cdots & \frac{1}{(d_{\ell}-3)!} \\
\vdots & \vdots & \vdots & \ddots & \vdots \\
0 & 0 & 0 & \cdots & 1
\end{pmatrix}.
\]
Thus the eigenvalues of $\cT_1$ are $\{e^{\lambda_{\ell}}\}_{\ell}$. Now let $X$ be an eigenvector of $\cT_1$ with eigenvalue $e^{\lambda_{\ell}}$ with $|e^{\lambda_{\ell}}| = 1$. We know from~\cite[Proposition 6.2]{Wolfe2012} that for such eigenvalues $d_{\ell} = 1$. Therefore, $X$ is also an eigenvector of $L$ with eigenvalue $\lambda_{\ell}$ and $\lambda_{\ell}$ is pure imaginary.

\end{proof}

A simple corollary is that when taking tensor powers of a channel, as is common in Shannon theory, the classical and quantum capacities are given by the zero-error capacities at infinite time.
  
It has been demonstrated that the zero-error capacities of quantum channel $\mathcal{T}: \mathcal{M}_d \to \mathcal{M}_d$ attain their infinite-time capacities after $d^2$ time concatenations\cite{singh2024-v1}. Consequently, for any QMS or infinitely divisible quantum channel on finite Hilber space, we have: 
\begin{proposition}\label{prop: zero-error finite time lindblad}
    Let \((\mc{T}_t)_{t > 0}\) be a quantum Markov semigroup. Then, for any \(\tilde{t} > 0\), we have:
    \[
    C_0(\mc{T}_{\tilde{t}}) = C_0^\infty(\mc{T}_1),
    \]
    \[
    Q_0(\mc{T}_{\tilde{t}}) = Q_0^\infty(\mc{T}_1).
    \]
\end{proposition}
\begin{proof}
    It is shown in~\cite{singh2024-v1}[Proposition 5.2] that any quantum channel \(\mc{T}: \mc{M}_d \to \mc{M}_d\) reaches its infinite-time classical and quantum zero-error capacities after \(t \geq d^2\). Specifically:
    \[
    C_0(\mc{T}^{d^2}) = C_0^\infty(\mc{T}) \quad \text{and} \quad Q_0(\mc{T}^{d^2}) = Q_0^\infty(\mc{T}).
    \]

    Now, for any \(\tilde{t} > 0\), we can express \(\mc{T}_{\tilde{t}}\) as:
    \[
    \mc{T}_{\tilde{t}} = \left( \mc{T}_{\tilde{t}/d^2} \right)^{d^2}.
    \]
    Using this decomposition, the zero-error classical and quantum capacities of \(\mc{T}_{\tilde{t}}\) are equivalent to the infinite-time zero-error capacities of \(\mc{T}_{\tilde{t}/d^2}\). By applying Proposition~\ref{prop: Markov semigroup capacities}, we have:
    \[
    C_0(\mc{T}_{\tilde{t}}) = C_0^\infty(\mc{T}_{\tilde{t}/d^2}) = C_0^\infty(\mc{T}_1),
    \]
    and
    \[
    Q_0(\mc{T}_{\tilde{t}}) = Q_0^\infty(\mc{T}_{\tilde{t}/d^2}) = Q_0^\infty(\mc{T}_1).
    \]

    This completes the proof.
\end{proof}

\section{Additivity of infinite-time capacities}

    \begin{lemma}\label{lem:peripheral projection of tensor product}
        Consider two quantum channels, $\mc T: \mc M_{d}\to\mc M_{d} $ and $\mc S:\mc M_{d'}\to \mc M_{d'}$, with their respective peripheral projections denoted by $\mc T_P$ and $\mc S_P$, and their respective peripheral subspaces denoted by $\chi_\mc T$ and $\chi_\mc S$. For the tensor product of these channels, $\mc T \otimes \mc S$, the peripheral projection is given by $\mc T_P \otimes \mc S_P$, and the corresponding peripheral subspace is $\chi_\mc T \otimes \chi_\mc S$.
    \end{lemma}

\begin{proof}
        By definition  there exist increasing sequences $\{m_i\}$, $\{n_i\}$, and $\{k_i\}$ such that
        $\lim_{i\to \infty} \mc T^{m_i} = \mc T_P$, $ 
        \lim_{i\to \infty} \mc S^{n_i} = \mc S_P$, and \quad
        $\lim_{i\to \infty} (\mc T\otimes \mc S)^{k_i} = \left(\mc T\otimes \mc S \right)_P.$
        Since any power of a peripheral projection is itself, we have
        \begin{align*}
            \lim_{i\to \infty} \mc T^{m_in_ik_i} &= \mc T_P,\\
        \lim_{i\to \infty} \mc S^{m_in_ik_i} &= \mc S_P, \\
        \text{and} \quad 
        \lim_{i\to \infty} (\mc T\otimes \mc S)^{m_in_ik_i} &= \left(\mc T\otimes \mc S \right)_P.
        \end{align*}
        
        Therefore, $\left( \mc T\otimes \mc S\right)_P = \mc T_P \otimes \mc S_P$.

        Next, we consider the peripheral subspaces. The peripheral subspace is the fixed-point subspace of the peripheral projection. If $X\in \chi_\mc T\otimes \chi_\mc S$, then it can be written as $X= \sum_i Y_i\otimes Z_i$ where $Y_i\in \chi_\mc T$ and $Z_i\in \chi_\mc S$. So we have
        \[\mc T_P\otimes \mc S_P (X)= \sum_i \mc T_P(Y_i)\otimes \mc S_P(Z_i)= \sum_i Y_i\otimes Z_i=X,\]
        showing that $X$ is a fixed point of  $\mc T_P\otimes \mc S_P$ . Therefore $\chi_\mc T \otimes \chi_\mc S \subseteq \chi_{\mc T\otimes \mc S}$.

        For the other direction, let $X\in \mc M_d\otimes \mc M_{d'}$ belongs to peripheral subspace of $\mc T\otimes \mc S$, and decompose $X$ as $X= \sum_i Y_i\otimes Z_i$  where $Y_i\in\mc M_d $ and $Z_i\in \mc M_{d'}$. Since $X$ is a fixed point of $\mc T_P \otimes \mc S_P$, we have 
        \[ X= \mc T_P \otimes \mc S_P(X)= \sum_i \mc T_P(Y_i) \otimes \mc S_P(Z_i).\]
        Because $\mc T_P(Y_i)\in \chi_\mc T$ and $\mc S_P(Z_i)\in \chi_\mc S$ ,  it follows that  $X\in \chi_\mc T \otimes \chi_\mc S  $. Thus  $\chi_{\mc T\otimes \mc S }\subseteq \chi_\mc T \otimes \chi_\mc S $. 
        Therefore, we conclude that $\chi_{\mc T\otimes \mc S }= \chi_\mc T \otimes \chi_\mc S$.
    \end{proof}
\subsection{Proof of Theorem~\ref{thm:additivity}}
\begin{proof}
        By Lemma~\ref{lem:peripheral projection of tensor product}, the  peripheral subspace of $\mc T\otimes \mc S $ is the tensor product of the peripheral subspace of $\mc T$ and $\mc S$. So if $\chi_\mc T= 0 \oplus \bigoplus_{k=1}^K \mc{M}_{d_k} \otimes \omega_k$ for the decomposition $\mb{C}^d = H_0 \oplus \bigoplus_{k=1}^K H_{k,1} \otimes H_{k,2}$ and $\chi_\mc S= 0 \oplus \bigoplus_{k'=1}^{K'} \mc{M}_{d'_{k'}} \otimes \omega'_{k'}$ for the decomposition $\mb{C}^{d'} = H'_0 \oplus \bigoplus_{k'=1}^{K'} H'_{k',1} \otimes H'_{k',2}$, then we can decompose
        \begin{align*}
            &\mb{C}^{d} \otimes \mb{C}^{d'} = \overline{H}_0 \oplus \\
            &\bigoplus_{k \in \{1, \dots, K\}, k' \in \{1, \dots, K'\}} H_{k,1} \otimes H'_{k',1} \otimes H_{k,2} \otimes H'_{k',2},
        \end{align*}

        where 
        \begin{align*}
        \overline{H}_0 = (H_0 &\otimes H'_0)\\
        &\oplus \left(H_0 \otimes (\bigoplus_{k'=1}^{K'} H'_{k',1} \otimes H'_{k',2}) \right) \\
        &\oplus \left((\bigoplus_{k=1}^K H_{k,1} \otimes H_{k,2}) \otimes H'_0\right)
        \end{align*}
        
        and get for this decomposition
        \[ \chi_{\mc T\otimes \mc S}= 0 \oplus \bigoplus_{k,k'} \mc{M}_{d_k\times d'_{k'}} \otimes \omega_k\otimes \omega'_{k'}.\]

        By Theorem~\ref{thm:classical and quantum capacity of repeated channel}, $C_{0}^{\infty}(\cT \otimes \mc{S}) = \log(\sum_{k,k'} d_k d'_{k'}) = C_{0}^{\infty}(\cT) + C_0^{\infty}(\mc{S})$ and $Q_{0}^{\infty}(\cT \otimes \mc{S}) = \log(\max_{k,k'} d_k d'_{k'}) = Q_{0}^{\infty}(\cT) + Q_0^{\infty}(\mc{S})$.
        
    \end{proof}

\subsection{Proof of Proposition~\ref{prop: Shannon theory capacity}}
 \begin{proof}
By Theorem~\ref{thm:classical and quantum capacity of repeated channel} and Theorem~\ref{thm:additivity}, we have
        \begin{align*} 
        C_\delta^\infty (\mc  T^{\otimes m})
        &= \log\left(\left\lfloor \frac{2^{C_0^\infty (\mc  T^{\otimes m})}}{1-\delta} \right\rfloor\right) \\
        &= \log\left(\left\lfloor \frac{2^{mC_0^\infty (\mc  T)}}{1-\delta} \right\rfloor\right).
        \end{align*}
        But 
        \begin{align*}
            m C_0^{\infty}(\cT) + \log \left(\frac{1}{1-\delta}\right) - 1 \leq 
            \log\left(\left\lfloor \frac{2^{m C_0^\infty (\mc  T)}}{1-\delta} \right\rfloor\right) \leq\\
            m C_0^{\infty}(\cT) + \log \left(\frac{1}{1-\delta}\right),
        \end{align*}
    which proves the desired result.

    For the quantum capacity, the same argument gives

        \begin{align*}
        m Q_0^{\infty}(\cT)+ \frac{1}{2} \log\left(\frac{1}{1-\delta}\right)-1 \leq  Q_{\delta}^\infty\left(\mc T^{\otimes m}\right) \\
            \leq m Q_0^{\infty}(\cT)+ \log\left(\frac{1}{1-\delta}\right), 
            \end{align*}
        which proves the desired result.  
        
    \end{proof}

    \section{Algorithm }

\begin{proof}[Proof of Theorem \ref{thm:algorithm-structure}]

    The first step is to compute the peripheral projection channel $\cT_{P}$. Let $\hat T \in \cM_{d^2}$ be the super-operator form of $\mc T$. Using the Jordan normal form, we can express $\hat{T}$ as
    
    \( \hat T = \sum_{\ell=1}^s \lambda_{\ell} P_\ell +N_\ell\), where $\lambda_{\ell}$ are the eigenvalues of $\hat{T}$, $P_{\ell}$ are projections and $N_{\ell}$ are nilpotent. The super-operator form of $\mc T_P$ is given by~\cite[Proposition 6.3]{Wolfe2012}
    \[ \hat T_P= \sum_{\ell: |\lambda_\ell|=1} P_\ell.\]
    Thus, $\hat{T}_P$ can be computed using the Jordan normal form of $\hat{T}$ which can be obtained in time $O((d^2)^3) = O(d^6)$.

    Recall that $\chi_{\cT} = \mathrm{Fix}(\cT_P)$, where $\mathrm{Fix}(\cS)$ is the fixed point space of a map $\cS$, i.e., $\mathrm{Fix}(\cS) = \mathrm{span}\{X \in \cM_{d} : \cS(X) = X \}$. The rest of the algorithm computes the structure of the fixed point space of the quantum channel $\cT_{P}$.

    Let $\chi_\mathcal{T} = \mathrm{Fix}(\cT_P) = 0 \oplus \bigoplus_{k=1}^K \mathcal{M}_{d_k} \otimes \omega_k$ for the Hilbert space decomposition 
    \begin{equation}\label{eq: hilbert decompose destortion}
        \mb C^d = H_0 \oplus \bigoplus_{k=1}^K H_{k,1}\otimes H_{k,2}.
    \end{equation}
    Instead of working directly with $\chi_{\mathcal{T}}$, it is convenient to work with the matrix algebra $\cA \subseteq B(\mb{C}^d)$ which is defined as
        \begin{equation} \label{eq: structure of von Neumann alg}
        \mathcal{A} = \bigoplus_{k=1}^K \mathcal{M}_{d_k} \otimes I_{d_k'}.
    \end{equation}
    for the Hilbert space decompostion~\eqref{eq: hilbert decompose destortion} with $d_k = \dim H_{k,1}$ and $d'_k = \dim H_{k,2}$. Then, $\chi_{\mathcal{T}}$ is called a distortion of the matrix algebra $\mathcal{A} = \bigoplus_{k=1}^K \mathcal{M}_{d_k} \otimes I_{d_k'}$, i.e., there is a completely positive map $\cD$ such that $\cD(\cA) = \chi_{\mathcal{T}}$.

    In~\cite[Theorem 5 and Section V]{blume2010information}, an algorithm is given to find operators $A_1, \dots, A_N \in \cM_{d}$ such that $\cA = \mathrm{span}(A_1, \dots, A_N)$. This is done by computing left and right eigenvectors of $\hat{T}$ and thus can be done in time $O(d^6)$.

    What remains is to determine the structure of $\mathcal{A}$. 
To do this, it is sufficient to find a complete set of basis vectors $|e_{k,i,j}\rangle$, where $k = 1, \dots, K$, $i = 1, \dots, d_k$, and $j = 1, \dots, d_{k'}$, such that any element of $\mathcal{A}$ has a block diagonal matrix form $0 \oplus \bigoplus_{k=1}^K A_k \otimes I_{d_{k}'}$. In other words, we need to satisfy the following condition:
\[
\langle e_{k',i',j'} | A | e_{k,i,j} \rangle = \delta_{k,k'} \delta_{j,j'} \langle i' | A_k | i \rangle.
\]
Note that the choice of this basis is not unique because we can choose any arbitrary basis for $\mathcal{M}_{d_k}$ and $\mathcal{M}_{d_{k'}}$, and form a basis for the algebra by taking their tensor product.
To construct such a basis, the process involves three steps:
\begin{enumerate}
    \item \textbf{Find Minimal Central Projections for $\mc A$:} First, identify orthogonal projectors $\mathcal{P}_k$ on the spaces $H_{k,1} \otimes H_{k,2}$.
    \item \textbf{Find Minimal Projection for $\mc A$:} Next, decompose each projector $\mathcal{P}_k$ as a sum of projections $P_{k,i}$, i.e., $\mathcal{P}_k = \sum_{i=1}^{d_k} P_{k,i}$, where $P_{k,i}$ is a projection onto a subspace of the form $|\psi_i\rangle \otimes H_{k,2}$, with $|\psi_i\rangle \in H_{k,1}$.
    \item \textbf{Construct the basis:} Finally, use these projections to construct the required set of basis vectors.
\end{enumerate}
This approach will yield the desired structure of the algebra $\mathcal{A}$.
   
    We start with the first step, i.e., computing $\mc P_k$. For that we use that the projectors $\mc P_k$ are the minimal projections in the center of $\cA$. Recall that the center $Z(\cA) = \{ X \in \cA : X Y = YX \; \forall Y \in \cA\}$ and that $Z(\cA) = \cA \cap \cA'$ where $\cA'$ is the commutant of $\cA$ defined by $\cA' = \{ X \in B(\mb{C}^d) : XY = YX \quad \forall Y \in \cA\}$. 

    Note that for the matrix algebra $\cA$ in~\eqref{eq: structure of von Neumann alg}, we have $\cA' = \bigoplus_{k=1}^K I_{H_{k,1}} \otimes \cM_{d_k'}$ and $Z(\cA) = \bigoplus_{k=1}^K \mb{C} I_{H_{k,1}} \otimes I_{H_{k,2}}$. Given an algebra $\cB$, a minimal projection in $\cB$ is an orthogonal projection $P$ such that $P \cB P = \mb{C} P$~\cite{holbrook2003noiseless}. The minimal projections of $Z(\cA)$, that are also called the minimal central projections, are exactly the projectors $\mc P_k$ we are looking for.
   
    Thus, in order to compute the projectors $\{\mc P_k\}_{k=1}^K$, we first compute a representation of $Z(\cA)$ as a linear span. We do this by computing a representation of the commutant $\cA'$ as the linear span of some operators $B_1, \dots, B_{N'}$, which can be done by computing the kernel of the matrix $\Gamma$ described in Lemma~\ref{lem: generators of commutant}. Then, the center can be computed by taking the intersection of subspaces $\cA$ and $\cA'$; see Algorithm~\ref{alg: central algebra}. Then Algorithm~\ref{alg: decomposing to minimal} described a general algorithm to find all minimal projections of an algebra in  $\mc M_d$ in time $O(N d^3 \log d)$, where $N$ is the number of operators describing the algebra as a linear span. Note that the center $Z(\cA)$ has dimension at most $d$ and so we may assume $N \leq d$. As such, we have computed all the $\{\mc P_k\}_{k=1}^K$ in time $O(d^5 \log d)$.

    Now, the second step is within each block $k$, we want to compute minimal projections $P_{k,i}$ such that $\sum_{i=1}^{d_k} P_{k,i}= \mc P_k$ where $d_k=\dim H_{k,1}$. For that we now compute minimal projections $P_{k,j}$ in $\cA$ satisfying $P_{k,j} \mc P_k = P_{k,j}$ in algebra $\cA$ by using Algorithm~\ref{alg: decomposing to minimal} with inputs $\mc P_k$ and $\{A_1, \dots, A_N\}$. For each $k$, the runtime for finding the minimal projections $\{P_{k,j}\}_{j=1}^k$ $O(d_k \times d^5 \log d)$ (as the dimension $N$ of the algebra $\cA$ can be up to $d^2$). As $\sum_{k=1}^K d_k \leq d$, the runtime of this step is $O(d^6 \log d)$.

    For the last step we will construct the basis by using $\mc P_k$ and $P_{k,i}$. We have that $d'_k = \rank(P_{k,j})$ and $d_k$ is the number of minimal projections found $\{P_{k,j}\}_j$. From the structure of the algebra $\cA$ (as given in \eqref{eq: structure of von Neumann alg}), we know that each minimal projection $P_{k,i}$ takes the form $|\alpha_{k,i}\rangle\langle \alpha_{k,i}| \otimes I_{H_{k,2}}$, where $\{|\alpha_{k,i}\rangle\}_{i=1}^{d_k}$ forms a complete basis for $H_{k,1}$.
    Thus, the support of $P_{k,i}$ (eigenvectors with unite eigenvalues) is of form of 
    \(\{ |\alpha_{k,i}\ra \otimes |\beta_{k,i,j}\ra\}\)
    , where $\{|\beta_{k,i,j}\ra\}_{j=1}^{d_k'} $ is a complete basis for $H_{k,2}$. One can find these set of vectors by computing the eigenvalue and eigenvectors of $P_{k,i}$. The last challenge is that although  $\{|\beta_{k,i,j}\ra\} $ is a complete basis of $H_{k,2}$ for any $i$, they are not  necessarily same for any two $i$ and $i'$. To overcome this problem, we should find the unitary maps $U_{k,m,n}$ such $|\beta_{k,m,j}\ra = U^{k,m,n} |\beta_{k,n,j}\ra$ for all $j$.
    
    For finding unitary map $U^{k,m,n}$, we define matrix $V^{k,m,n}$ for $A\in \mc A$ as 
    \[ V^{k,m,n}_{i,j} := \la \alpha_{k,n}|\la \beta_{k,n,j} |\mc P_k A \mc P_k | \alpha_{k,m} \ra|\beta_{k,n,i}\ra.\]
    As $A$ has the structure in form of $\bigoplus_{k=1}^K A_k \otimes I_{H_{k,2}}$, we have $$V^{k,m,n}_{i,j}= \la \alpha_{k,n} |A_k |\alpha_{k,m} \ra \times \la \beta_{k,n,j}|\beta_{k,m,i}\ra. $$
    If $\la \alpha_{k,n} |A_k |\alpha_{k,m} \ra$ is non-zero, then 
    \[ U^{k,m,n}= \frac{V^{k,m,n}}{\tr(V^{k,m,n^\dagger}V^{k,m,n})}.\]
    Since $\mc A$ is spanned by $\{A_i\}$, for any $k$, $m$, and $n$, there exists at least one $A \in \{A_i\}$ such that $V^{k,m,n}$ is a non-zero matrix . This ensures that we can construct all of the $U^{k,m,n}$.By using $U^{k,m,n}$, we can construct the basis in the form of 
     \begin{align*}
            |e_{k,i,j}\ra :&= |\alpha_{k,i}\ra \otimes |\beta_{k,1,j}\ra \\
            &= \sum_m \la \beta_{k,i,m}|\beta_{k,1,j}\ra |\alpha_{k,i,m}\ra|\beta_{k,i,m} \ra \\
            &= \sum_m U^{k,1,i}_{j,m} |\alpha_{k,i}\ra |\beta_{k,i,m}\ra.
     \end{align*}
\end{proof}
 
Below we provide a lemma that was used in the above proof. For this lemma we need the concept of operator-vector correspondence: given a matrix $X\in \cM_d$, it represents a vector,  $\left|X\right\rangle\rangle\in \mathbb{C}^n\otimes \mathbb{C}^n$.

    \begin{lemma}\label{lem: generators of commutant}
            Let \(\mathcal{A}\) be a matrix algebra generated by $\{A_1, \cdots, A_N\}$, and \(\mathcal{A}'\) be the commutant of \(\mathcal{A}\). Then \(X\) belongs to \(\mathcal{A}'\) if and only if \(\left|X\right\rangle\rangle\) belongs to the kernel of \(\Gamma\), where 
            \begin{equation}\label{eq: Gamma mc A}
                \Gamma = \sum_{i=1}^N \left( A_i \otimes I - I \otimes A_i^T \right)^\dagger \left( A_i \otimes I - I \otimes A_i^T \right).
            \end{equation}
    \end{lemma}

 \begin{algorithm}[t]
    \scriptsize
    \caption{Find the structure of $\chi_\mathcal{T}$}
    \label{alg: Find the structure of chi}
    \begin{algorithmic}[1]
        \State \textbf{Input:} $\hat{T}$, the super-operator form of $\mathcal{T}$
        \State \textbf{Output:} The structure of $\chi_\mathcal{T}$ as in~\eqref{eq:peripheral-decomposition}
        \Procedure{Peripheral-Subspace-Structure}{$\cT$}
            \State $\hat{T} = \sum_{\ell=1}^s \lambda_{\ell} P_{\ell} + N_{\ell}$ \Comment{Jordan decomposition}
            \State $\hat{T}_P \gets \sum_{\ell : |\lambda_{\ell}| = 1} P_{\ell}$
           
            \State \textbf{return} \textsc{Fixed-Subspace-Structure}($\hat{T}_P$) \Comment{Algorithm~\ref{alg: Find the structure of fixed points}}
        \EndProcedure
    \end{algorithmic}
\end{algorithm}

    \begin{algorithm}[t]
    \scriptsize
    \caption{Find the fixed point structure $\mathrm{Fix}(\mathcal{S})$}
    \label{alg: Find the structure of fixed points}
    \begin{algorithmic}[1]
        \State \textbf{Input:} $\hat{S}$, the super-operator form of $\mathcal{S}$
        \State \textbf{Output:} The structure of $\mathrm{Fix}(\mathcal{S})$
        \Procedure{Fixed-Subspace-Structure}{$\mathcal{S}$}
            \State $\{A_1, \ldots, A_N\} \gets \textsc{Per-algebra-as-linear-span}(S)$ \Comment{\cite{blume2010information} repr. of $\cA$ as a linear span}
            \State $\{C_1, \ldots, C_{M}\} \gets \textsc{Center-of-algebra}(A_1, \ldots, A_N)$ \Comment{Algorithm~\ref{alg: central algebra}}
            \State Ensure $\{A_j\}_j$ and $\{C_j\}_{j}$ are Hermitian via $A \to A + A^{\dagger}$ and $A \to i(A - A^\dagger)$
            \State $P \gets \text{projector on support of the center of $\cA$}$ \Comment{$P$ projector onto $\oplus_{k=1}^K H_{k,1} \otimes H_{k,2}$}
            \State minimalCentralProj \(\gets\) \textsc{Minimal-projections}$(P, \{C_i\})$ \Comment{Algorithm~\ref{alg: decomposing to minimal}}
            \For{$k \gets 1$ \textbf{to} $K$} 
                \State $P_k \gets$ minimalCentralProj[$k$]
                \State minimalProj[$k$] $\gets$ \textsc{Find-minimal-projections}$(P_k, \{A_i\})$ \Comment{Algorithm~\ref{alg: decomposing to minimal}} 
            \EndFor
            \State{ basisSet$\gets$ \textsc{Construct-Basis}($ \{A_i\}$, minimalCentralProj, minimalProj )}\Comment{Algorithm~\ref{alg: Construct the Basis}}
            \State Compute $\{\omega_k\}$ as in~\cite[Lemma 5.4]{blume2010information}
            \State \textbf{return} basisSet, $\{ \mc P_k\}$
        \EndProcedure
    \end{algorithmic}
\end{algorithm}

    \begin{algorithm}[t]
        \scriptsize
        \caption{Compute center of algebra $\mathcal{A} = \text{span} \{A_1, \ldots, A_N\}$}
        \label{alg: central algebra}

        \begin{algorithmic}[1]
            \State \textbf{Input:} $A_1, A_2, \ldots, A_N \in \cM_d$ \Comment{$\cA = \mathrm{span}(A_1, \dots, A_N)$}
            \State \textbf{Output:} $\{C_1, \ldots, C_M\}$ that span the center of $\mathcal{A}$
            \Procedure{Center-of-algebra}{$A_1, A_2, \ldots, A_N$}
                \State $\Gamma \gets  \sum_{i} (A_i^\dagger A_i) \otimes I 
                    - A_i^\dagger \otimes A_i^T
                    - A_i \otimes \overline{A}_i
                    + I \otimes (\overline{A}_i A_i^T)$ \Comment{$\bar{A}$ is the complex conjugate of $A$}  
                \State Compute the kernel of $\Gamma$ as $\text{span}\{ B_1, \ldots, B_{N'}\}$ \Comment{See Lemma~\ref{lem: generators of commutant}, $\text{span}\{ B_1, \ldots, B_{N'}\} = \cA'$}
                \State $\Gamma' 
                    \gets \sum_{i} (B_i^\dagger B_i) \otimes I 
                    - B_i^\dagger \otimes B_i^T
                    - B_i \otimes \overline{B}_i  
                    + I \otimes (\overline{B}_i B_i^T)$   \Comment{$\ker \Gamma'$ corresponds to $\cA'' = \cA$}
                \State Compute the kernel of $\Gamma + \Gamma'$ as $\text{span}\{ C_1, \ldots, C_M\}$ \Comment{$\Gamma, \Gamma' \geq 0$, so $\ker \Gamma + \Gamma' = \ker \Gamma \cap \ker \Gamma'$}
                \State \textbf{return} $\{C_1, \ldots, C_M\}$ 
            \EndProcedure \Comment{The runtime of this procedure is $O(N d^4 + d^6)$}
        \end{algorithmic}
    \end{algorithm}

    \begin{algorithm}[t] 
            \scriptsize

        \caption{Find  projection smaller than $P$}
        \label{alg: reduce P}
        \begin{algorithmic}[1]
            \State \textbf{Input:} Orthogonal projection $P$ and $\{A_i\}_{i=1}^N$ spanning $\cA$
            \State \textbf{Output:} Projection $Q \in \cA$ such that $QP = Q$, and $\tr(Q) \leq \tr(P)/2$ if $P$ is not minimal
            \Procedure{Reduce-Projection}{$P, \{A_i\}$ }
                \If{$\exists i$ is such that $P A_i P \not\in \mb{C} P$} \Comment{$P$ is not minimal}
                    \State Write spectral decomposition $P A_i P = \sum_{j} \lambda_j P_j$ \Comment{At least $2$ nonzero eigenvalues}
                    \State \Comment{\cite{holbrook2003noiseless} shows that $P_j \in \cA$ for all $j$}
                    \State \textbf{return} $Q = \argmin\{\tr(P_j)\}$ \Comment{We have $\tr(Q) \leq \tr(P)/2$ as $\sum_{j} P_j \leq P$}
                \Else \Comment{$P$ is minimal}
                    \State \textbf{Return } $P$
                    \EndIf
            \EndProcedure \Comment{The runtime of this procedure is $O(N d^3)$}
        \end{algorithmic}
    \end{algorithm}

    \begin{algorithm}[t]
            \scriptsize

            \caption{Find one minimal projection in the range of projection \(P\)}
            \label{alg: find one minimal in P}
            \begin{algorithmic}[1] \label{alg:find one minimal in P}
                \State \textbf{Input:} Orthogonal projection $P$ and $\{A_i\}$ spanning $\cA$
                \State \textbf{Output:} A minimal projection $Q \in \cA$ such that $Q \leq P$
                \Procedure{Find-one-minimal-projection}{$P, \{A_i\}$}
                    \State $Q\gets$ \textsc{Reduce-Projection}($P, \{A_i\}$) \Comment{Algorithm~\ref{alg: reduce P}}
                    
                    \While{$P\neq Q$} \Comment{As the trace is divided by $2$ at each step, at most $\log d$ steps} 
                        \State $P\gets Q$
                        \State $Q\gets$ \textsc{Reduce-Projection}($P, \{A_i\}$) 
                    \EndWhile
        
                    \State \textbf{return} $P$ 
                \EndProcedure \Comment{\textsc{Reduce-Projection} is called at most $\log d$ times}
            \end{algorithmic}
        \end{algorithm}

    \begin{algorithm}[t]
            \scriptsize 

        \caption{Decomposing $P$ to minimal projection in algebra $\mc A$}
        \label{alg: decomposing to minimal}
        \begin{algorithmic}[1]
            \State \textbf{Input:} Orthogonal projection $P \in \cA$ and $\{A_i\}$ spanning $\cA$
            \State \textbf{Output:} A set of minimal projections $Q_1, \dots, Q_{s}$ in $\cA$ such that $Q_1 + \dots + Q_{s} = P$
            \Procedure{Find-minimal-projections}{$P,\{A_i\}$}
                
                \State minimalProjections \(\gets \{\}\) 
    
                \While{$P \neq 0$} \Comment{Iterate until $P$ becomes zero}
                    \State $Q \gets \textsc{Find-one-minimal-projection}(P)$ \Comment{Algorithm~\ref{alg: find one minimal in P}}
                    \State Add $Q$ to minimalProjections 
                    \State $P \gets P - Q$ \Comment{Update $P$ by removing the minimal projector that was found}
                    \State \Comment{Note that $P-Q$ is also an orthogonal projector}
                \EndWhile
    
                \State \textbf{return} minimalProjections 
            \EndProcedure
        \end{algorithmic}
    \end{algorithm}

\begin{algorithm}[t]
\scriptsize
\caption{Construct the Basis \( e_{k,i,j} \) for \( \mathcal{A} \)}
\label{alg: Construct the Basis}

\begin{algorithmic}[1]
    \State \textbf{Input}: $\{A_i\}$ spanning $\mathcal{A}$, Orthogonal Minimal Central Projections \(\{ P_k \}\), Orthogonal Minimal Projections \( \{ P_{k,i} \} \) for each \( P_k \)
    \State \textbf{Output}: The basis vectors \( e_{k,i,j} \)
    \Procedure{Construct-Basis}{$ \{A_i\}, \{ P_k \}, \{ P_{k,i} \} $}
        \State $K \gets$ number of minimal central projections $\{ P_k \}$
        \State $d[k] \gets$ number of minimal projections $\{P_{k,i}\}$
        \State $\text{d-prime}[k] \gets \operatorname{tr}(P_{k,1})$

        \Comment{Step 1: Compute Eigenvectors for Minimal Projections}
        \For{each \(P_{k,i}\) where \(k=1, \dots, K\) and \(i=1, \dots, d[k]\)}
            \State $\text{eigenVectors}[k][i] \gets$  of eigenvectors of \( P_{k,i} \) with unit eigenvalues 
        \EndFor

        \Comment{Step 2: Construct matrices \( U^{k,1,n} \)}
\For{$k = 1, \dots, K$ and $n = 1, \dots, d[k]$}
    \State Initialize $V \gets 0$
    \While{ $V=0$}\Comment{Exit loop once a valid \( V \) is found}
    \State Pick $A$ from $\{A_i\}$
        
        \For{$i ,j = 1, \dots, \text{d-prime}[k]$}
            \State $V[i,j] \gets \text{eigenVectors}[k][n][j]^\dagger \cdot A \cdot \text{eigenVectors}[k][1][i]$
        \EndFor
       
    \EndWhile
    \State $U[k][n] \gets \frac{V}{\operatorname{tr}(V^\dagger V)}$ \Comment{Normalize \( U[k][n] \)}
\EndFor

        \Comment{Step 3: Construct the basis vectors \( e_{k,i,j} \)}
\For{$k = 1, \dots, K$, $i = 1, \dots, d[k]$, and $j = 1, \dots, \text{d-prime}[k]$}
    \State Initialize $\text{basisVector} \gets 0$
    \For{$m = 1, \dots, \text{d-prime}[k]$}
        \State $\text{basisVector} \gets \text{basisVector} + U[k][i][j,m] \times \text{eigenVectors}[k][i][m]$
    \EndFor
    \State $\text{basisSet}[k][i][j] \gets \text{basisVector}$ \Comment{Store in $\text{basisSet}[k][i][j]$}
\EndFor
        \State \textbf{Return} The set \( \text{basisSet} \)
    \EndProcedure
\end{algorithmic}
\end{algorithm}

    \begin{proof}
        By the fact that $A\otimes B|C\ra\ra = |ACB^T\ra\ra$ (see~\cite[Proposition 2.20]{Watrous_2018}) applies if the dimensions of $A$, $B$, and $C$ indicate $ACB^T$ is a valid matrix. So $X\otimes I |Y\ra\ra= |XY\ra\ra$ and $I\otimes X^T|Y\ra\ra =|YX\ra\ra$. Therefore if we define 
        $$\Gamma_{A_i}:=\left( A_i \otimes I - I \otimes A_i^T \right)^\dagger \left( A_i \otimes I - I \otimes A_i^T \right),$$
         then an operator $X$ commutes with $A_i$ if and only if $\Gamma_{A_i}|X\ra\ra= 0$. As $\Gamma_{A_i}$ is positive,  the kernel of $\Gamma$ is the intersection of kernels of all $\Gamma_{A_i}$. Thus, $X \in \cA'$ if and only if $| X \ra \ra \in  \Gamma$.
         
    \end{proof}

\begin{remark}
    The structure of the peripheral (fixed) subspace of collective noise for different numbers of qubits is summarized in Table~\ref{tab:peripheral_subspace}. Additionally, the time required to compute this structure using the code provided in~\cite{mostafa_taheri_2024_12821315} on an Apple M2 Pro processor is included.

    \begin{table*}[t]
    \centering
    \begin{tabular}{|c|c|c|c|}
    \hline
    $n$ & Dimension & $\chi_\mc{T}$ & Time \\ \hline
    3 & $8 \times 8$ & $CI_4 \oplus (\mc M_2 \otimes I_2)$ & 18.576 ms  \\ \hline
    4 & $16 \times 16$ & $\mc M_2\oplus CI_5\oplus(\mc M_3 \otimes I_3)$ & 647.805 ms  \\ \hline
    5 & $32 \times 32$ & $CI_6 \oplus (\mc M_5 \otimes I_2) \oplus (\mc M_4 \otimes I_4)$ & 4.763 s \\ \hline
    6 & $64 \times 64$ & $\mc M_5 \oplus CI_7 \oplus (M\mc M_5 \otimes I_5) \oplus(\mc M_9 \otimes I_3)$ & 223 s \\ \hline
    \end{tabular}
    \caption{Structure of $\chi_\mc{T}$, matrix sizes, and computation times for various number of qubit.}
    \label{tab:peripheral_subspace}
    \end{table*}
\end{remark}
\section{Infinite dimension}
Before commencing the proof of Proposition~\ref{prop: infinte-hilbert-capacity}, we present a proposition, and then we proceed to outline the proof.
\begin{lemma}[Theorem 19 of~\cite{olkiewicz_1999}]\label{lemma:structure of isometric subspace}

    Let \(\mathcal{T}: \operatorname{Tr}(H) \to \operatorname{Tr}(H)\) be a CPTP map and contractive in the operator norm. Then the isometric subspace \(\Lambda_{\mathcal{T}}\) is a \(\mathcal{T}\)-invariant subspace and \(\Lambda_{\mathcal{T}}\) decomposes as
\[
\Lambda_{\mathcal{T}} = \bigoplus_{k} \Lambda_{\mathcal{T}}^{(k)},
\]
    where, for all \(\phi \in \Lambda_{\mathcal{T}}^{(k)}\) and \(\psi \in \Lambda_{\mathcal{T}}^{(l)}\), we have \(\tr(\psi \phi) = 0\) for \(k \neq l\).
For each \(k\), there exists a Banach space isomorphism $\alpha_k :\Lambda_\mc T^{(k)}\to \Tr(\tilde{ H}_k)$, such that the action of \(\mathcal{T}\) on \(\Lambda_{\mathcal{T}}^{(k)}\) corresponds to \(U^\dagger_k \cdot U_k\), where \(U_k\) is a unitary operator on \(\tilde{H}_{k}\), i.e. 
\[ \alpha_k \circ \mc T \circ \alpha_k^{-1} (.) = U_k \: . \: U_k^\dagger\]
\end{lemma}

\begin{remark}
    If \(\Lambda_\mc{T}\) is not finite, then the infinite-time classical capacity of \(\mc{T}\) will be infinite. However, the infinite-time quantum capacity is not necessarily infinite. For example, consider the quantum channel \(\mc{T}: \Tr(H) \to \Tr(H)\) defined as:
    \[
    \mc{T}(X) = \sum_{i \in \mathbb{Z}} K_i X K_i^\dagger,
    \]
    where \(K_i = |i\rangle\langle i-1|\), and \(\{ |i\rangle \}\) is the orthonormal basis of \(H\). In this case, the infinite-time classical capacity is infinite (\(C_0^\infty(\mc{T}) = \infty\)), but the infinite-time quantum capacity is zero (\(Q_0^\infty(\mc{T}) = 0\)).

\end{remark}

\begin{proof}[Proof of Proposition~\ref{prop: infinte-hilbert-capacity}]
   
We begin by showing that the isometric and peripheral subspaces are the same if \(\Lambda_\mc{T}\) is finite-dimensional. Afterward, we define a restricted channel \(\bar{\mc{T}}\) on a finite-dimensional Hilbert space that acts similarly to \(\mc{T}\) on the isometric subspace. Finally, we use the result from Theorem~\ref{thm:classical and quantum capacity of repeated channel} to determine the achievable capacity of \(\mc{T}^t\) in the limit as \(t \to \infty\).

For any CPTP map \(\mc{T}\), the peripheral subspace is always a subset of the isometric subspace (see Proposition 9 of~\cite{olkiewicz_1999}). To prove \(\Lambda_\mc{T} = \chi_\mc{T}\), it suffices to show that \(\Lambda_\mc{T} \subseteq \chi_\mc{T}\) when \(\Lambda_\mc{T}\) is spanned by a finite number of generators.

By Lemma ~\ref{lemma:structure of isometric subspace}, the isometric subspace \(\Lambda_\mc{T}\) can be decomposed into \(K\) subspaces \(\{\Lambda_\mc{T}^{(k)}\}_{k=1}^K\), where each \(\Lambda_\mc{T}^{(k)}\) is isomorphic to \(\Tr(\tilde{H}_k)\), with \(\tilde{H}_k\) a finite-dimensional Hilbert space. Let \(\{|i_k\ra\}_{i=1}^{d_k}\) be a complete basis of \(\tilde{H}_k\), where \(d_k = \dim(\tilde{H}_k)\), such that the unitary operator \(U_k\) acts as \(U_k |i_k\ra = e^{i\theta_i} |i_k\ra\), for some phases \(\theta_i\). Considering the action of \(U_k\) on the operators \(|i_k\ra\la j_k|\), we have:
\[
U_k |i_k\ra\la j_k| U_k^\dagger = e^{i(\theta_i - \theta_j)} |i_k\ra\la j_k|.
\]

 Using the isomorphism \(\alpha_k^{-1}\) that maps \(\Tr(\tilde{H}_k)\) to \(\Lambda_\mc{T}^{(k)}\), we find:
\begin{align*}
    \mc{T} \circ \alpha_k^{-1}(|i_k\ra\la j_k|) &= \alpha_k^{-1}(U_k |i_k\ra\la j_k| U_k^\dagger) \\
    &= e^{i(\theta_i - \theta_j)} \alpha_k^{-1}(|i_k\ra\la j_k|).
\end{align*}

This shows that \(\alpha_k^{-1}(|i_k\ra\la j_k|)\) belongs to the peripheral subspace \(\chi_\mc{T}\) for all \(i, j\). Since \(\Tr(\tilde{H}_k)\) is spanned by the operators \(\{|i_k\ra\la j_k|\}\), it follows that \(\Lambda_\mc{T}^{(k)}\) is spanned by \(\{\alpha_k^{-1}(|i_k\ra\la j_k|)\}\). Therefore, \(\Lambda_\mc{T}^{(k)} \subseteq \chi_\mc{T}\). Summing over all \(k\), we conclude that \(\Lambda_\mc{T} \subseteq \chi_\mc{T}\). Combined with the fact that \(\chi_\mc{T} \subseteq \Lambda_\mc{T}\), we conclude that \(\Lambda_\mc{T} = \chi_\mc{T}\).

Since we assumed $\Lambda_\mc T$ is spanned by a finite number of generators, and  we know that $\Lambda_\mc T$ is generated by finite-dimensional orthogonal projections~\cite[Proposition 5]{olkiewicz_1999},therefore, $\Lambda_\mc T$ is generated by a finite number of finite-dimensional  projections, i.e

\begin{equation}
\label{eq: lambda-T span projections}
    \Lambda_\mc T = \operatorname{span} \{ P_i\in \Tr( H); i=1,\cdots d\}.
\end{equation} 
Let us define $S = \sum_{i=1}^d P_i$. 
We let $H'$ be the support of $S$. As $H'$ is finite-dimensional, there exists a subspace $H''$ such that $H = H' \oplus H''$.
So any $X\in \Lambda_\mc T$ can be written as $x\oplus 0$ in    the decomposition $H' \oplus H''$. 

 Now let us show that the map $\mc{T}$ keeps $\Tr( H')\oplus 0$ invariant.  Let  $x \in \Tr( H')$ be positive semi-definite operator, and $S=s\oplus 0$.  Recall that $x$ and $s$ are both  positive semi-definite operator on finite-dimensional Hilbert space and $\supp(x)\subseteq \supp(s)$, thus there exists $\epsilon>0$ such that $s- \epsilon x>0$~\cite[Lemma 1.1]{blume2010information}. From Lemma~\ref{lemma:structure of isometric subspace} $\Lambda_\mc T$ is invariant under action of $\mc T$, so $\mc T(s\oplus 0)= s' \oplus 0$. As $\mc T$ is linear we have 
\[ \mc T(s\oplus 0) - \epsilon \mc T(x \oplus 0)= s'\oplus0 -\epsilon(\mc T(x \oplus 0))  \geq 0.\]
Since $\mc T$ is complete positive, both $s'\oplus0$ and $\mc T(x \oplus 0)$ are positive semi-definite. Consequently,  the support of $\mc T(x \oplus 0)$ must be a subset of  the support of $s'\oplus0$, thus $\supp(\mc T(x\oplus 0)) \subseteq  H' $.  It can be easily generalized to any $x\in \Tr(H')$ by using the fact that any operator can be written as linear combination of four positive semi-definite operator.
Therefore $\mc T$ keeps $\Tr(H')$ invariant. 

Let us now define the map $\bar{\mc{T}} : \Tr(H') \to \Tr(H')$ as follows: For any $x \in \Tr(H')$, let $\bar{\mc{T}}(x) = \mc{T}(x \oplus 0)$, where $x \oplus 0$ is written in the decomposition $H' \oplus H''$. Since $\Tr(H')$ is invariant under the action of $\mc T$, it is easy to verify that $\bar{\mc T}$ is a CPTP map on operators on the finite-dimensional Hilbert space $H'$.

Now we can apply Eq.~\eqref{eq:peripheral-decomposition-1} and Theorem~\ref{thm:classical and quantum capacity of repeated channel} to $\bar{\mc{T}}$.  By definition, the peripheral subspace of $\bar{\mc T}(X)$ is also $\Lambda_{\mc{T}}$ (up to $0$ acting on $H''$).

This gives the achievable infinite-time capacity of \(\mc{T}\). The only remaining task is to clarify how the encoder and decoder on \(\Tr(H')\) can be extended to \(\Tr(H)\). For an encoder $\mc E: \Tr(\mathbb{C}^d)\to \Tr(H')$, we define the extension to $\Tr(H)$ as $\tilde{\mc E}(X) =\mc E(X) \oplus 0$. For a decoder $\mc{D} : \Tr(H') \to \Tr(\mathbb{C}^d)$, we define the extension to $\Tr(H)$ as $\tilde{\mc{D}}(X) = \mc{D}(P_{H'} X P_{H'}) + (1-\tr(P_{H'}X)) \proj{\psi}$ where $\ket{\psi}$ is an arbitrary state in $H''$ and $P_{H'}$ is the orthogonal projector on $H'$.
\end{proof}

\end{document}